\documentclass[12pt]{iopart}
\expandafter\let\csname equation*\endcsname\relax
\expandafter\let\csname endequation*\endcsname\relax
\usepackage{amsmath}
\usepackage{amssymb}
\usepackage{graphicx}
\usepackage{epsfig}
\usepackage{algorithm}
\usepackage{epstopdf}
\usepackage{subfigure}
\usepackage{multirow}
\numberwithin{equation}{section}
\usepackage [pdftex,
 bookmarksnumbered,
 bookmarksopen,
 colorlinks,
 anchorcolor=blue,
 citecolor=green,
 linkcolor=red
 ]{hyperref}
\begin{document}
\newtheorem{Thm}{Theorem}[section]
\newtheorem{Cor}[Thm]{Corollary}
\newtheorem{Lem}[Thm]{Lemma}
\newtheorem{Prop}[Thm]{Proposition}
\newtheorem{Def}[Thm]{Definition}
\newtheorem{Rem}[Thm]{Remark}
\newtheorem{Assump}[Thm]{Assumption}

\newtheorem{Exam}[Thm]{Example}
\newenvironment{proof}[1][Proof]{\noindent\textbf{#1.} }{\ \rule{0.5em}{0.5em}}

\title[Dynamic Thresholding Algorithm]{Dynamic Thresholding Algorithm with Memory for Linear Inverse Problems}

\author{Zhong-Feng Sun$^1$, Yun-Bin Zhao$^2$, Jin-Chuan Zhou$^1$ and  Zheng-Hai Huang$^3$}

\address{$^1$School of Mathematics and Statistics, Shandong
University of Technology, Zibo, Shandong,  China.}
\address{$^2$Corresponding author. Shenzhen International Center for Industrial and Applied Mathematics, SRIBD, The Chinese University of Hong Kong, Shenzhen,
 China. }
\address{$^3$School of Mathematics, Tianjin University, Tianjin, China. }
\ead{zfsun@sdut.edu.cn, yunbinzhao@cuhk.edu.cn, jinchuanzhou@sdut.edu.cn and  huangzhenghai@tju.edu.cn}
\vspace{10pt}
\begin{indented}
\item[](First version April 1, 2024; Revised, Nov. 11, 2024)
\end{indented}

\begin{abstract}
The relaxed optimal $k$-thresholding pursuit (ROTP) is a recent algorithm for linear inverse problems. This algorithm is based on the optimal $k$-thresholding technique which performs vector thresholding and error metric reduction simultaneously. Although ROTP can be used to solve small to medium-sized linear inverse problems, the computational cost  of this algorithm is high when solving large-scale problems. By merging the optimal $k$-thresholding technique and iterative method with memory as well as  optimization with sparse search directions, we propose the so-called dynamic thresholding algorithm with memory (DTAM), which iteratively and dynamically selects vector bases to  construct the problem solution. At every step, the algorithm uses more than one or all iterates generated so far to construct a new search direction, and solves only the small-sized quadratic subproblems at every iteration. Thus the computational complexity of DTAM is remarkably lower than that of ROTP-type methods.
It turns out that DTAM can locate the solution of linear inverse problems if the matrix involved satisfies the restricted isometry property.  Experiments on synthetic data, audio signal reconstruction and image denoising demonstrate that the proposed algorithm performs comparably to several mainstream thresholding and greedy algorithms, and it works much faster than the ROTP-type algorithms especially when the sparsity level of signal is relatively low.\\

\end{abstract}


\section{Introduction}
 A typical linear inverse problem is to reconstruct unknown data $ d\in \mathbb{R}^n$ via some linear measurements $y \in \mathbb{R}^m $ subject to noise effects:
\begin{align}\label{LIP}
y=Bd+\nu,
\end{align}
where $B \in \mathbb{R}^{m\times n} $ is a given measurement matrix with  $m \ll n $,  and $\nu \in \mathbb{R}^m $ is a noise vector. This problem arises in many scenarios, where the number of measurements $m$ is much smaller than the length of the target vector $d.$  For instance, when using CT for medical diagnosis, it is expected to use as little radiation dose as possible in order to reduce the impact of radiation on the patient. Also, in the same and many other application scenarios, the target signal often admits certain special structure that makes it possible to reconstruct the signal from the underdetermined system (\ref{LIP}). In fact, many natural signals and images can be sparsely represented under some orthogonal linear transforms (e.g.,  discrete wavelet transforms). As a result, we may assume that the target data $d$ can be represented as  $d=\Phi^T x$, where $\Phi\in \mathbb{R}^{n\times n} $ is a transform matrix and the vector $x\in\mathbb{R}^{ n}$ is sparse (or compressible in the sense that it can be approximated by a sparse vector). In such cases, reconstructing $d$  via solving the linear inverse problem \eqref{LIP} amounts to recovering a sparse (or compressible) vector $x$ through the following system:
\begin{align}\label{SLIP}
y=Ax+\nu,
\end{align}
where $A=B\Phi^T \in \mathbb{R}^{m\times n}$ is still called the measurement matrix. As the solution $x$ of this problem is sparse, the problem above can be referred to as a sparse linear inverse problem.  This problem has a wide range of applications in such areas as image processing \cite{LAZS20,XDY21}, wireless
communication \cite{BSR17,CZJAZ21,LFL18}, sensor networks  \cite{CW16C,CW16O}, to name a few. The system (\ref{SLIP}) can be reformulated as the sparse optimization problem
 \begin{align}\label{model}
\underset{x\in \mathbb{R}^n}{\min}\{ {\left\| y-Ax\right\|}_2^2: \left\| x\right\|_0\leq k\},
\end{align}
where $ k$  is a given  integer number reflecting the sparsity level of $x,$  and $ \left\| \cdot\right\|_0$ denotes the number of nonzero entries of a vector. For the convenience of discussion, we list the main abbreviations used in the paper in Table \ref{tab-abbr}.
\begin{center}
   \begin{table}[hbtp]
 \caption{\label{tab-abbr}List of Abbreviations}
 \vspace{0.2cm}
 \begin{indented}
 \item[]
 \begin{tabular}{ll}
  \hline
 Abbreviation& Full Name\\
 \hline
DTAM & Dynamic thresholding algorithm with memory\\
DWT&  Discrete wavelet transform\\
 EDOMP&Enhanced dynamic orthogonal matching pursuit \cite{ZL23}\\
gOMP&Generalized orthogonal matching pursuit  \cite{WKS12}\\
NTP& Natural thresholding pursuit \cite{ZL22} \\
OMP& Orthogonal matching pursuit \cite{E10,TG07}\\
PGROTP&Partial gradient relaxed
optimal $k$-thresholding pursuit  \cite{MZKS22} \\
PSNR&Peak signal-to-noise ratio \\
RIC& Restricted isometry constant\\
RIP&Restricted isometry property \\
ROTP&Relaxed optimal $k$-thresholding pursuit  \cite{Z20}\\
ROTP$\omega$ & ROTP with $\omega$  times of data compression at each iteration \cite{Z20}\\
 SNR&Signal-to-noise ratio\\
SP& Subspace pursuit \cite{DM09}\\
StOMP&Stagewise
orthogonal matching pursuit  \cite{DTDS12}\\
\hline
 \end{tabular}
  \end{indented}
 \end{table}
\end{center}

 Thresholding is a large class of widely used algorithms for sparse optimization problems \eqref{model}. This class of algorithms includes the hard thresholding \cite{BD08,BD10,F11,KC14,MZ20,SZZM22,WHZ23}, optimal $k$-thresholding \cite{MZ22,MZKS22,SZZ23,Z20,ZL21}, soft thresholding \cite{ BT09,BSR17,DDM04,
 D95,E06,LZC16,YOG08}, and the recent natural thresholding pursuit \cite{ZL22}.
Although the hard thresholding selecting indices of a few  largest magnitudes of a vector can guarantee the iterates generated by the algorithm are  feasible to \eqref{model}, it is generally not an optimal thresholding approach from the viewpoint of minimizing the error metric $\|y-Ax\|_2^2,$  as pointed out in \cite{Z20}. Thus a more sophisticated data compression method called the  optimal  $k$-thresholding was first introduced in  \cite{Z20}, based on which the family of optimal $k$-thresholding algorithms, termed ROTP$\omega,$  were  proposed in \cite{Z20}, where  $\omega$ reflects the times of data compression in every iteration. Although ROTP$\omega$  is generally more stable and robust for solving linear inverse problems than  hard thresholding and greedy algorithms \cite{Z20,ZL21},  its computational cost  remains high since the algorithm needs to solve  quadratic optimization subproblems in the course of iteration. To reduce the cost, some modifications of ROTP$\omega$ using acceleration or linearization  techniques have been proposed  recently \cite{GLMNQ23,MZKS22,SZZ23,ZL22}. For instance,   PGROTP \cite{MZKS22} and the heavy-ball-based ROTP \cite{SZZ23} were developed by incorporating the partial gradient and heavy-ball acceleration into ROTP (ROTP$\omega$  with $\omega=1$), respectively.  Numerical results  indicate that PGROTP can be faster than  ROTP2 \cite{MZKS22}.  However, PGROTP is still time-consuming when solving large-scale problems. It is worth mentioning that the  natural thresholding pursuit (NTP)   in \cite{ZL22}, using linearization of quadratic subproblem,  remarkably reduces the complexity of ROTP-type algorithms. In addition, the stochastic counterpart of NTP was recently developed in \cite{GLMNQ23} for sparse optimization problems.

Except for thresholding algorithms,  the greedy methods are also a popular class of  algorithms for solving sparse linear inverse problems. OMP is one of such greedy algorithms \cite{E10,TG07}  which gradually identifies the support of solution to the problem by selecting  only one index in each iteration.  The index selected by OMP corresponds to  the  largest absolute component of the gradient of error metric, i.e., the objective function in (\ref{model}).
The OMP and its modified versions were analyzed in such references as \cite{CW14,DW10,Mo15}. However, theoretical and numerical results indicate that OMP tends to be inefficient as  the sparsity level $k$ becomes large.  The main reason for this might be that when $k$ is relatively large and when the large magnitudes are close to each other, there is no guarantee for a correct index being selected by the OMP procedure, and many significant indices corresponding to large magnitudes in gradient are completely discarded at every iteration. This means most useful information conveyed by the gradient of the current iterate is ignored in OMP procedure. Motivated by this observation, several modifications of OMP with different index selection criteria were introduced, including gOMP \cite{WKS12}, StOMP \cite{DTDS12}, EDOMP \cite{ZL23} and SP \cite{DM09}.  For instance, at every iteration, gOMP picks a fixed number,  $K,$ of the largest magnitudes of gradient.  However, such a selection rule might   result in a wrong index set especially when the gradient is $s$-sparse with $s<K$ since in such a case the algorithm have to pick more indices than necessary. On the contrary, StOMP and EDOMP adopt certain dynamic index selection criteria whose purpose is to efficiently use the information of significant gradient components.  EDOMP is generally stable, robust and efficient for sparse signal recovery, although the convergence of EDOMP has not yet established at present \cite{ZL23}.

Inspired by the dynamic index selection strategies in StOMP \cite{DTDS12} and EDOMP \cite{ZL23} and iterative methods with memory \cite{ABD23,LZG23,NSH21},  we propose a new algorithm called dynamic thresholding algorithm with memory (DTAM) in this paper. The algorithm is different from existing ones in three aspects:  (i) The iterative search direction in this method is a combination of the gradients of more than one or all iterates generated so far by the algorithm instead of the only gradient for the current iterate. (ii) The index selection in this algorithm is dynamic according to a rule defined by a generalized  mean function \cite{ZFL06} evaluated at the current search direction with memory.  It should be pointed out that the generalized mean function is used for the first time to serve such a purpose. (iii)  The algorithm adopts a novel dimensionality reduction strategy based on the sparsity of  iterative point and search direction. The key idea here is to reduce a  high-dimensional  quadratic optimization problem to a low-dimensional one whose dimension is at most twice of the sparsity level of the solution to the linear inverse problem. We also carry out a rigorous analysis of DTAM to establish an error bound which measures the distance between the solution of the problem  and iterates generated by the algorithm.  The error bound is established under  the  restricted isometry property (RIP). It implies that DTAM is guaranteed to locate the $k$-sparse solution of linear inverse problem  if the matrix satisfies the RIP of order $3k.$ Moreover, as a byproduct of our analysis, the convergence of PGROTP  with $\bar{q}=k$ is also obtained in this paper for the  first time, which is given in Corollary \ref{Cor-PG}. The numerical performances of DTAM and several existing algorithms including  PGROTP \cite{MZKS22}, NTP \cite{ZL22}, StOMP  \cite{DTDS12}, SP \cite{DM09} and OMP \cite{E10,TG07} are compared through experiments on threes types of sparse linear inverse problems:  The problems with synthetic data, practical audio signal reconstruction and image denoising. Numerical results indicate that the proposed algorithm  does perform very well for solving linear inverse problems compared with several existing algorithms, and it works faster than PGROTP.

The paper is organized as follows. In Section \ref{sec-prelim},  we introduce some useful inequalities,  generalized mean functions, the PGROTP algorithm, and the new algorithm DTAM. The analysis of DTAM is performed in Section \ref{convergence}. Numerical results are reported in Section \ref{Num-exp}, and the conclusions are given in last section.

\section{Preliminary and Algorithms}\label{sec-prelim}

Some notations that will be used in the paper are summarized in Table \ref{table-notation}.
\begin{center}
\begin{table}[htbp]
 \renewcommand\arraystretch{1.1}
\caption{List of Notations}\label{table-notation}
 \vspace{0.2cm}
 \begin{indented}
 \item[]
 \begin{tabular}{ll}
 \hline
 Notation& Definition\\
\hline
$\mathbb{R}_{++}^n$& The positive orthant of $\mathbb{R}^n$\\
$N$ & Index set $ \{1,2,\ldots,n\} $\\
  $|\Omega|$ &Cardinality of the set $\Omega\subseteq N$\\
 $\overline{\Omega}$ &Complement set  of $\Omega\subseteq N$, i.e., $\overline{\Omega}=N\setminus\Omega$  \\
$ \textrm{supp} (u)$ &Support of $u\in\mathbb{R}^n$, i.e., $\textrm{supp} (u)=\{i\in N:u_i\neq 0\}$\\
$u_\Omega$& $n$-dimensional vector obtained from $u\in \mathbb{R}^n$ with entries
  $(u_\Omega)_i=u_i$ \\
  & for $ i\in \Omega$
  and $(u_\Omega)_i=0$ for $i\in \overline{\Omega}$\\
$ |u|$ & Absolute
value of the vector $u\in \mathbb{R}^n$, i.e., $|u|=(|u_1|,\ldots,|u_n|)^T$\\
$\mathcal{L}_k(u)$&Index set of the  $k$ largest entries in magnitude of $u\in\mathbb{R}^n$ \\
$\mathcal{H}_k(u)$& Hard thresholding  of $u\in\mathbb{R}^n$, i.e., $\mathcal{H}_k(u)=u_{\Omega}$ where $\Omega=\mathcal{L}_k(u)$ \\
 $\|\cdot\|_i$&$\ell_i$-norm of a vector, $1\leq i\leq+\infty$\\
 $u \circ  v$& Hadamard product of $u$ and $v$ in $\mathbb{R}^n$, i.e., $u \circ v=(u_1v_1,\ldots,u_nv_n)^T$\\
 {\bf e} &The vector of ones in $\mathbb{R}^n$, i.e., {\bf e}~$ =(1,\ldots,1)^T$\\
  \hline
 \end{tabular}
  \end{indented}
 \end{table}
\end{center}

 \subsection{Basic inequalities}\label{sec-inq}
 Let us  first recall the RIC and RIP of an $m\times n$ matrix $A$ with $m<n.$
\begin{Def} \emph{\cite{CT05}} \label{def-RIC}
Given a matrix $A\in \mathbb{R}^{m\times n}$ with $m< n$. The $k$th order RIC of $A$, denoted by
$\delta_k,$ is the smallest nonnegative number  $\delta$ which obeys
\begin{align}\label{def-RIC-1}
(1-\delta){\left\| x \right\|}^2_2\leq  {\left\| Ax \right\|}^2_2\leq (1+\delta){\left\| x \right\|}^2_2
\end{align}
for all $k$-sparse vectors $x\in \mathbb{R}^n$. Moreover, the matrix $A$ is
said to satisfy the RIP of order $k$ if $\delta_k<1$.
\end{Def}

The following  lemma is very useful for the analysis of DTAM.
\begin{Lem} \emph{\cite{F11,Z20}} \label{lem-basic-ineq}
Let $u \in \mathbb{R}^n$ and $ v \in \mathbb{R}^m$ be two vectors, $s \in N$ be a positive  integer and $W\subseteq N$ be an index set.
\begin{itemize}
\item [\rm (i)]   If  $ |W\cup  \textrm{supp} (u)|\leq s$, then
$\| \left ( (I-A^ T A)u\right)_W\|_2\leq\delta_s {\| u \|}_2.$
\item [\rm (ii)]   If  $ |W|\leq s$, then
$\| \left ( A^ Tv\right)_W\|_2\leq\sqrt{1+\delta_s} {\| v \|}_2.$
\end{itemize}
\end{Lem}

The next fundamental property of orthogonal projection can be found in \cite[Eq.(3.21)]{F11} and \cite[pp.49]{Z20}.
\begin{Lem}\emph{\cite{F11,Z20}} \label{lem-pursuit}
Let $x \in \mathbb{R}^n$ be a vector satisfying $y = Ax+\nu$ where $ \nu\in R^m$ is a noise vector.   Let $ \Omega \subseteq N$ be an index set satisfying $|\Omega|\leq k$ and
\begin{equation*}
u^*=\arg \min_{u\in \mathbb{R}^n} \{ {\| y-Au\|}_2^2:  \textrm{supp} (u)\subseteq \Omega \}.
\end{equation*}
Then
\begin{equation*}
\|u^*-x_S\|_2
\leq \frac{1}{\sqrt{1-(\delta_{2k})^2}}\| (x_S)_{\overline{\Omega}}\|_2+\frac{ \sqrt{1+\delta_{k}}}{1-\delta_{2k}}\|\nu' \|_2,
\end{equation*}
where $S:=\mathcal{L}_k(x)$ and $\nu':=\nu+Ax_{\overline{S}}. $
\end{Lem}

We  also need the following result taken from \cite[Lemma 4.1]{ZL21} concerning hard thresholding.
\begin{Lem} \emph{\cite{ZL21}}\label{lem-hard-thres}
Let $u,h\in\mathbb{R}^n$ be two vectors and  $\| h\|_0\leq k$. Then
\begin{equation*}
\|h-\mathcal{H}_k(u)\|_2\leq \|(u-h)_{S\cup S^*}\|_2+\|(u-h)_{S^*\setminus S}\|_2,
\end{equation*}
where $S:= \textrm{supp} (h)$ and $S^*:= \textrm{supp} (\mathcal{H}_k(u))$.
\end{Lem}

\subsection{Generalized mean function}\label{GMF}

A generalized mean function defined in \cite{ZFL06} will be used in the algorithm proposed in next section. Let us state a result for generalized mean functions, which can be obtained directly
from \cite[Theorem 2.7]{ZFL06}.

\begin{Lem}\label{lem-GM-fun}
Let $\Lambda$ be an open interval in $\mathbb{R}$ and $[0,1]\subset \Lambda$, and let  $\theta\in \mathbb{R}_{++}^k$ be a given positive vector. Let $\Psi, \phi_i :\Lambda\rightarrow \mathbb{R},$   $i=1,\ldots,k$ be strictly increasing and twice continuously differentiable, and let $\Psi$ be convex and $\phi_i, i=1,\ldots,k $ be strictly convex. Assume that there exist constants $\tau$ and $\tau_i>0, i=1,\ldots,k$ such that  for $ t\in \Lambda $
\begin{align*}%\label{phi}
\tau_i \phi_i(t)\phi_i^{''}(t)\geq [\phi_i^{'}(t)]^2 , ~~\tau \Psi(t)\Psi^{''}(t)\leq [\Psi^{'}(t)]^2.
\end{align*}
If $\tau\geq \max_{1\leq i\leq k}\tau_i$, then the generalized mean function
\begin{align}\label{G-M-fun}
\Gamma_\theta(z)=\Psi^{-1}\left(\sum_{i=1}^k \theta_i \phi_i(z_i)\right),
\end{align}
 where $ z =(z_1, ..., z_k)^T \in \Lambda^k,$ is convex, strictly increasing and twice continuously differentiable.
\end{Lem}

If $\Psi=\phi_1=\cdots=\phi_k$, \eqref{G-M-fun} reduces to the mean function in \cite{HLP34}. Some specific examples of generalized mean functions satisfying Lemma \ref{lem-GM-fun} are given as follows.

\begin{Exam}\label{example}
\emph{(i)} Taking $\Lambda=\mathbb{R},$ constant $\sigma>0$ and $\Psi(t)=\phi_i(t)=e^{t/\sigma} $ $ (i=1,\ldots,k), $   we get
\begin{align}\label{Exam1-G-M}
\Gamma_\theta(z)=\sigma\ln \left(\sum_{i=1}^k \theta_i e^{ z_i /\sigma}\right),~ ~  z\in \Lambda^k.
\end{align}

\emph{(ii)} Taking $\sigma>0,$  $\Lambda=(-\sigma,+\infty)$ and   $l>1$  as well as $\Psi(t)=\phi_i(t)=(t+\sigma)^l $ $ (i=1,\ldots,k), $ one has
\begin{align*}
\Gamma_\theta(z)= \left(\sum_{i=1}^k \theta_i ( z_i+\sigma)^l\right)^{1/l}-\sigma,~ ~  z\in \Lambda^k.
\end{align*}

\emph{(iii)} Taking $\sigma>0, $ $\Lambda=(-\sigma,+\infty),$     and $\Psi(t)=(t+\sigma)^l$ with $1<l\leq 2$ and $\phi_i(t)=\Delta_{1,1}(t+\sigma) \textrm{ or } \Delta_{1,2}(t+\sigma)$ ($ i=1,\ldots,k$),
where $\Delta_{1,1}$ and $ \Delta_{1,2} $ are functions defined as
$$
\Delta_{1,1}(\hat{t})={\hat{t}}^2/2-{\hat{t}}+\ln({\hat{t}}+1),~~\Delta_{1,2}({\hat{t}})=\frac{1}{2}\left[({\hat{t}}+1)^2-({\hat{t}}+1)^{-1}-3{\hat{t}}\right]
$$
with $\hat{t}\in(0,+\infty),$ one has
\begin{align*}%\label{Exam2-G-M}
\Gamma_\theta(z)= \left(\sum_{i=1}^k \theta_i \phi_i( z_i)\right)^{1/l}-\sigma,~ ~  z\in \Lambda^k.
\end{align*}
\end{Exam}

\subsection{Algorithms}\label{sec-alg}
Before we state our algorithm, let us first recall the PGROTP algorithm in \cite{MZKS22}, which uses the partial gradient  to speed up the ROTP.
\begin{algorithm}
\caption{~Partial gradient relaxed
optimal $k$-thresholding pursuit  (PGROTP) }
Input the data $(A,y)$, integer number $\bar{q}\geq k$ and initial point $x^0$.
\begin{itemize}
\item[ S1.] At the current point $x^p$, set
 $
u^p=x^p+\mathcal{H}_{\bar{q}}(A^ T(y-Ax^p)).
$
\end{itemize}

\begin{itemize}
\item[ S2.]Generate the index set $S^{p+1}=\mathcal{L}_k(u^p \circ {w^p})$ by solving the optimization problem
\begin{align}\label{algorithm-ROTP-2}
w^p=\arg \min_{w\in \mathbb{R}^n} \{ {\| y-A(u^p \circ {w})\|}_2^2:  ~ \sum_{i=1}^n w_i=k, ~ 0\leq w \leq  {\bf e}\}.
\end{align}
\end{itemize}

\begin{itemize}
\item[ S3.]
Compute the next iterate $x^{p+1}$ by solving the orthogonal projection problem
 \begin{align*}%\label{algorithm-ROTP-3}
x^{p+1}=\arg \min_{x\in \mathbb{R}^n}  \{ {\| y-Ax\|}_2^2: \textrm{supp}(x)\subseteq S^{p+1}\}.
\end{align*}
\end{itemize}
Repeat S1-S3 until a certain stopping criterion  is met.
\end{algorithm}
 However, PGROTP still has computational complexity similar to that of ROTP$\omega$ \cite{MZKS22,ZL21}. It is worth mentioning that  the convergence of PGROTP was shown only for the case  $\bar{q}\geq 2k$ at present \cite{MZKS22}.

 The dynamic index selection rules in StOMP \cite{DTDS12} and EDOMP \cite{ZL23} aim to efficiently use the information provided by the gradient-based search direction to predict the problem solution. The iterative  method with memory aims to use more than one or all generated iterates to obtain a search direction. Moreover, as shown in PGROTP, using part of the search direction may help lower the dimension of quadratic optimization subproblem in ROTP-type method. Thus by merging these techniques, we propose the so-called dynamic thresholding algorithm with memory (DTAM) for linear inverse problems, in which a new dynamic index selection strategy based on the following generalized mean function is adopted:
\begin{equation} \label{meanfff}  f(z):=\Gamma_{\theta}(z)-\Gamma_{\theta}(0), \end{equation}
where $\Gamma_{\theta}(z)$  is given by \eqref{G-M-fun}.

\begin{algorithm}\label{alg-DSROTP}
\caption{~Dynamic Thresholding Algorithm with Memory (DTAM)}
Input data $(A,y,k)$ and the parameters $\gamma\in(0,1]$ and $\beta \in [0,1).$ Input a generalized mean function of the form (\ref{meanfff}) with given parameter $\theta\in \mathbb{R}_{++}^k.$ Set the initial point $x^0=0$.
\begin{itemize}
\item[ S1.] Let  $r^p=\sum_{j=0}^p\beta^{p-j}\hat{r}^j$, where $\hat{r}^j=A^ T(y-Ax^j)$ for $j=0,\ldots,p$. Let  $\Omega_i =\mathcal{L}_i(r^p),~ i=1,\ldots, k.$ Set
 \begin{equation}\label{SDir} u^p=x^p+ (r^p)_{\Omega_q}, \end{equation}  where $q$ is determined as
 \begin{align}\label{alg-DSROTP-1}
q=\min\left\{i: ~ f\left(\frac{|r^p_{(i,k)}|}{\|r^p_{(k,k)}\|_2}\right)\geq\gamma f\left(\frac{|r^p_{(k,k)}|}{\|r^p_{(k,k)}\|_2}\right),~i=1,\ldots,k\right\},
\end{align}
in which  $r^p_{(i,k)},$ $i=1, ...,k$ are $k$-dimensional vectors whose entries are those of $(r^p)_{\Omega_i}$ supported on $\Omega_k$, i.e.,
\begin{align*}
 r^p_{(i,k)}=\{\left((r^p)_{\Omega_i}\right)_j: j\in  \Omega_k\}.
\end{align*}
\end{itemize}

\begin{itemize}
\item[S2.]   Let $V^p=\textrm{supp}(x^p)\cup \Omega_q $. If $|V^p|\leq k$, set $S^{p+1}=V^p$. Otherwise, if $|V^p|>k$, set $S^{p+1}=\mathcal{L}_k(u^p \circ {w^p})$, where $w^p$ is the solution to the problem
\begin{align}\label{alg-DSROTP-2}
 \min_{w\in \mathbb{R}^n} \left\{ {\| y-A(u^p \circ {w})\|}_2^2: ~ w_j=0\textrm{ for }j\notin V^p, \sum_{i\in V^p}w_i=k, ~0\leq w \leq  {\bf e}\right\}.
\end{align}
\end{itemize}

\begin{itemize}
\item[ S3.] Let
 \begin{align}\label{alg-ROTP-3}
x^{p+1}=\arg \min_{x\in \mathbb{R}^n} \{ {\| y-Ax\|}_2^2: ~ \textrm{supp}(x)\subseteq S^{p+1}\}.
\end{align}
\end{itemize}
Repeat S1-S3 until a certain stopping criterion is met.
\end{algorithm}

At the first step S1 of DTAM,  the vector $r^p = \hat{r}^p + \beta \hat{r}^{p-1} + \cdots + \beta^p \hat{r}^{0}  $ is the combination of negative gradients of $\|y-Ax\|_2^2/2 $ at the generated iterates. As the coefficients  $\beta^{\ell}, \ell=0, , ..., p$ are decaying as $\ell$ increases,  a more recent iterate is allocated a weight larger than their predecessors. For this reason, $\beta$ is referred to as a forgetting factor.   When $\beta=0,$ the vector $r^p$ reduces to the negative gradient at the current point $x^p$. The search direction $(r^p)_{\Omega_q}$ adopted in (\ref{SDir}) is a hard thresholding of $r^p$, i.e., $ (r^p)_{\Omega_q}= {\cal H}_q (r^p).$ The number $q$ is uniquely determined by the index selection rule (\ref{alg-DSROTP-1}) which is dynamically changed during iteration. Note that the inequality (\ref{alg-DSROTP-1}) is always satisfied for $i=k,$ and hence there exists a smallest  $i$ such that the inequality holds. Since $u^p=x^p+ (r^p)_{\Omega_q}$, we have $\textrm{supp}(u^p)\subseteq\textrm{supp}(x^p)\cup \Omega_q$. When $\textrm{supp}(u^p) \not= \textrm{supp}(x^p)\cup \Omega_q,$  performing the relaxed optmal $k$-thresholding of $ u^p$ with index set $ V^p=\textrm{supp}(x^p)\cup \Omega_q$ (i.e., solving the convex optimization problem (\ref{alg-DSROTP-2}))  might reduce the objective function in (\ref{alg-DSROTP-2}) more than the case $ V^p=\textrm{supp}(u^p) $.   It is well known that such reduction might help speed up the algorithm and enhance the convergence of the algorithm. There are several choices of the stopping criteria for DTAM. For instance, we may stop the algorithm after being performed a prescribed number of iterations, or we may stop the algorithm when $ \|y-Ax^{p}\|_2\leq \varepsilon,$ where $ \varepsilon$ is a given tolerance.

\begin{Rem} \label{Complexity} \emph{We can compare the computational complexity of DTAM and ROTP-type algorithms.
 The complexity of  ROTP$\omega$ and  PGROTP  in each iteration is $O(mn+m^3 +n^{3.5}L_n)$ \cite{MZKS22,ZL21}, where $L_n$ depending on $n$ is the size of the problem data encoding in binary \cite{MS90}. The main cost of the ROTP-type algorithms  is solving the quadratic optimization problem \eqref{algorithm-ROTP-2} which requires $O(n^{3.5}L_n)$ flops based on an interior-point algorithm \cite{MS90,T88,Y87}. It is evident that the actual dimension of \eqref{alg-DSROTP-2} is $ |V^p|$ which is at most $2k$, and thus solving \eqref{alg-DSROTP-2} only requires $O(k^{3.5}L_{2k})$ flops. Therefore, the complexity of DTAM with a simple generalized mean function is about $O(mn+m^3+k^{3.5}L_{2k})$ in each iteration, which is much lower than that of ROTP-type algorithms.}
\end{Rem}

\begin{Rem} \label{Mot-GMF} \emph{A big difference of DTAM from related existing methods lies in the index selection rule  in which the general mean function is used. The purpose of using the generalized mean function is to provide a relatviely more general framework of the algorihtm so that the  theoretical result can be established in  broader settings and more alternative index selection rules can be used for  implementation. To see how  the choice of generlized mean functions might influence the performance of the algorihtm, let us first establish the solution error bound in the next Section and then make a further discussion in Remark \ref{Rem-Zhao} on this issue. While the DTAM using general mean functions for index selection instead of the thesholding rule as in EDOMP\cite{ZL23} (to which  the error bound of EDOMP has not yet  established so far in the literature), we can establish the solution error bound (including the convergence) of DTAM under suitable conditions, as shown in Theorem \ref{main-thm} in the next section.}
\end{Rem}

\section{Error bound of DTAM} \label{convergence}
The purpose of this section is to establish the solution error bound of DTAM under the RIP of order $3k$. In other words, we show the convergence of DTAM via estimating the distance between the solution of linear inverse problem and the iterates generated by DTAM.
First, we need to establish several technical results. The first one displays the relation of $\|(r^p)_{\Omega_q}\|_2$ and $\|(r^p)_{\Omega_k}\|_2$, which is essential to bound the term $\|(u^p-x_S)_S\|_2  $ in order to eventually obtain the main result in this section.

\begin{Lem}\label{lem-rp_Omg_qk} Let $f(z) =  \Gamma_{\theta}(z)-\Gamma_{\theta}(0)$  where $\Gamma_{\theta}(z) $ is a generalized mean function satisfying Lemma \ref{lem-GM-fun}. Let $\gamma,  q, r^p, \Omega_q $ and $\Omega_k$  be given as in DTAM.   Then
\begin{align}\label{rp_Omg_qk}
\|(r^p)_{\Omega_q}\|_{2}\geq g(\gamma)\|(r^p)_{\Omega_k}\|_2,
\end{align}
where
\begin{align}\label{g-gam}
g(\gamma)=
\frac{2\gamma c }{\sqrt{\|\nabla f(0)\|_2^2+2\gamma c \lambda_*}+\|\nabla f(0)\|_2} <1
\end{align}
with $$c:=\min_{1\leq i\leq k} \frac{\partial f}{\partial z_i}(0)>0, ~ \lambda_*:=\max_{z\in [0,1]^k}\lambda_{max}(z)\geq 0,$$ where $\lambda_{max}(z)$ is  the largest eigenvalue of the Hessian matrix $\nabla^2f(z)$.
\end{Lem}

\begin{proof}
It follows from Lemma \ref{lem-GM-fun} that $f(z)=\Gamma_{\theta}(z)-\Gamma_{\theta}(0)$ is strictly increasing, twice continuously differentiable and convex in $\Lambda^k \supseteq [0,1]^k.$ Hence, the largest eigenvalue $\lambda_{max}(z)$ of the Hessian matrix $\nabla^2f(z)$ is continuous in $\Lambda^k$. It is easy to check that
\begin{align}\label{property-f}
f(0)=0, ~~\frac{\partial f}{\partial z_i}(0)>0\textrm{ for }1\leq i\leq k,~~\nabla^2f(z)\succeq 0\textrm{ for }z\in \Lambda^k.
\end{align}
Therefore, $$ c=\min_{1\leq i\leq k} \frac{\partial f}{\partial z_i}(0)>0, ~   \lambda_*=\max_{z\in [0,1]^k}\lambda_{max}(z)\geq 0. $$
Let $r^p_{(i,k)}$, $ i=1, ..., k$ be defined as in DTAM and denote by $s:=1/\|r^p_{(k,k)}\|_2. $
By the Taylor expansion, there exists $\xi\in[0,1]^k$ such that
\begin{align}\label{f-Omg-q}
f(s|r^p_{(q,k)}|)&=f(0)+s|r^p_{(q,k)}|^T\nabla f(0) +\frac{s^2}{2}|r^p_{(q,k)}|^T\nabla^2f(\xi)|r^p_{(q,k)}|\nonumber\\
&\leq s\|\nabla f(0)\|_2\|r^p_{(q,k)}\|_2+\frac{s^2}{2}\lambda_*\|r^p_{(q,k)}\|_2^2.
\end{align}
On the other hand, since
  $f(z)$ is convex, it follows from \eqref{property-f} that
\begin{align}\label{f-Omg-k}
f(s|r^p_{(k,k)}|)\geq f(0)+s|r^p_{(k,k)}|^T \nabla f(0) \geq sc  \|r^p_{(k,k)}\|_1\geq sc  \|r^p_{(k,k)}\|_2=c.
\end{align}
From \eqref{alg-DSROTP-1}, we have $f(s|r^p_{(q,k)}|)\geq \gamma f(s|r^p_{(k,k)}|)$. This together with  \eqref{f-Omg-q} and \eqref{f-Omg-k} implies that
$$
 \frac{s^2}{2}\lambda_*\|r^p_{(q,k)}\|_2^2 +s\|\nabla f(0)\|_2\|r^p_{(q,k)}\|_2-\gamma c\geq 0.
$$
By setting $\widetilde{t}= s\|r^p_{(q,k)}\|_2$ which is less than or equal to 1,  the above inequality is written as
\begin{align*}
 \frac{\lambda_*}{2}(\widetilde{t})^2+\|\nabla f(0)\|_2\widetilde{t}-\gamma c\geq 0.
\end{align*}

 {\bf Case 1.}  $\lambda_*=0.$ In this case, the above inequality reduces to $ \|\nabla f(0)\|_2\widetilde{t}-\gamma c\geq 0,$ i.e., $ \widetilde{t} \geq \frac{\gamma c}{\|\nabla f(0)\|_2}=  g(\gamma)$ for this case. Thus the inequality \eqref{rp_Omg_qk}  holds in this case.

 {\bf Case 2.}  $\lambda_*> 0.$  Since $\|\nabla f(0)\|_2\geq \min_{1\leq i\leq k} \frac{\partial f}{\partial z_i}(0)=c$,  we see that $0<\gamma\leq 1< (\lambda_*+2\|\nabla f(0)\|_2)/(2c)$ under which the quadratic equation  $\frac{\lambda_*}{2}t^2+\|\nabla f(0)\|_2t-\gamma c=0$  has a unique positive root $g(\gamma) $ in $(0,1)$ given as \eqref{g-gam}. This implies that  $\widetilde{t}\in [g(\gamma),1]$ which is exactly the inequality \eqref{rp_Omg_qk} by noting that $\|r^p_{(q,k)}\|_2=\|(r^p)_{\Omega_q}\|_{2}$ and $\|r^p_{(k,k)}\|_2=\|(r^p)_{\Omega_k}\|_{2}. $
\end{proof}

We now estimate  the upper bound of $\|(u^p-x_S)_S\|_2$   which is used to establish the   error bound for DTAM, as shown in Theorem \ref{main-thm}.
\begin{Lem}\label{bound_up-x} Let $x \in \mathbb{R}^n$ satisfy that $y = Ax+\nu$ where $ \nu$ is a noise vector. Denote by $S=\mathcal{L}_k(x)$ and $\nu':=y-Ax_S$. Then the  vectors $ u^p$ and  $x^j,~j=0,\ldots,p$   generated by DTAM satisfy that
\begin{align}\label{bound-up-xS}
\|(u^p-x_S)_S\|_2 \leq C_1 Q_p+\beta Q_{p-1}+\frac{C_2}{1-\beta}\|\nu'\|_2,
\end{align}
where
\begin{align}\label{Qp}
Q_i:=\sum_{j=0}^i\beta^{i-j}\|x^j-x_S\|_2 \emph{\textrm{ with }} i=p-1, p
\end{align}
and $ C_1$ and $ C_2$ are constants given as
\begin{align}\label{C12}
C_1=\sqrt{2}\delta_{3k}+\sqrt{1-[g(\gamma)]^2}(1+\delta_{3k}),~C_2=\sqrt{1+\delta_{2k}}\left(\sqrt{2}+\sqrt{1-[g(\gamma)]^2}\right),
\end{align}
where $g(\gamma)\in (0,1)$ is given by \eqref{g-gam}.
\end{Lem}

\begin{proof}
From the definition of $u^p$  in \eqref{SDir}, we have
\begin{align}\label{up-x-bd1}
\|(u^p-x_S)_S\|_2&=\|(x^p-x_S+ r^p)_{S}-( r^p)_{S\setminus\Omega_q}\|_2\nonumber\\
&\leq\|(x^p-x_S+ r^p)_{S}\|_2+\|(r^p)_{S\setminus {\Omega_q}}\|_2,
\end{align}
where the equality follows from the fact that $S\cap \Omega_q=S\setminus(S\setminus {\Omega_q})$.
Since $r^p=\sum_{j=0}^p\beta^{p-j}\hat{r}^j$ and  $S\setminus {\Omega_q}=(S\setminus {\Omega_k})\cup[(\Omega_k\setminus \Omega_q)\cap S]$,  the terms on the right hand of \eqref{up-x-bd1} can be bounded  as
\begin{align}\label{up-x-bd01}
\|(x^p-x_S+ r^p)_{S}\|_2
&=\left\|\sum_{j=0}^p\beta^{p-j}(x^j-x_S+\hat{r}^j)_{S}-\sum_{j=0}^{p-1}\beta^{p-j}(x^j-x_S)_{S}\right\|_2\nonumber\\
&\leq\sum_{j=0}^p\beta^{p-j}\|(x^j-x_S+\hat{r}^j)_{S}\|_2+\sum_{j=0}^{p-1}\beta^{p-j}\|x^j-x_S\|_2
\end{align}
and
\begin{align}\label{up-x-bd02}
\|(r^p)_{S\setminus {\Omega_q}}\|_2\leq\|(r^p)_{S\setminus {\Omega_k}}\|_2+\|(r^p)_{(\Omega_k\setminus \Omega_q)\cap S}\|_2\leq\|(r^p)_{S\setminus {\Omega_k}}\|_2+\|(r^p)_{\Omega_k\setminus \Omega_q}\|_2.
\end{align}
Since $\Omega_k=\mathcal{L}_k(r^p)$ and $|S|=k$, we get
$\|(r^p)_S\|_2^2\leq\|(r^p)_{\Omega_k}\|_2^2$. Eliminating the contribution of $S\cap\Omega_k$, we have
\begin{align}\label{up-x-bd03}
\|(r^p)_{S\setminus\Omega_k}\|_2\leq\|(r^p)_{\Omega_k\setminus S}\|_2.
\end{align}
From S3 in DTAM, we see that $x^j$ is the  solution of the quadratic optimization problem
 $$\min_{x\in \mathbb{R}^n} \{ {\| y-Ax\|}_2^2: ~ \textrm{supp}(x)\subseteq S^{j}\}$$ for $j=1,\ldots,p$. Thus the first-order optimality condition implies that $ (\hat{r}^j)_{S^{j}}=0,~j=1,\ldots,p,$
where $\hat{r}^j$ represents the negative gradient of  $\|y-Ax\|_2^2/2 $ at $x^j$. Since  $\textrm{supp}(x^j)\subseteq S^{j}$ for $j=1,\ldots,p$ and $x^0=0,$ we claim that  $\textrm{supp}(x^j)\cap \textrm{supp}(\hat{r}^j)=\emptyset$ for $j=0,\ldots,p$, i.e.,
\begin{align}\label{relation-r-x}
 (\hat{r}^j)_{\textrm{supp}(x^j)}=0,~(x^j)_{\textrm{supp}(\hat{r}^j)}=0,~~j=0,\ldots,p,
\end{align}
which implies that
\begin{align}\label{up-x-bd3-1}
\|(r^p)_{\Omega_k\setminus S}\|_2
&=\left\|\sum_{j=0}^p\beta^{p-j}(\hat{r}^j)_{\Omega_k\setminus S}\right\|_2
=\left\|\sum_{j=0}^p\beta^{p-j}(\hat{r}^j)_{\textrm{supp}(\hat{r}^j)\cap\Omega_k\setminus S}\right\|_2\nonumber\\
&=\left\|\sum_{j=0}^p\beta^{p-j}(x^j-x_S+\hat{r}^j)_{\textrm{supp}(\hat{r}^j)\cap\Omega_k\setminus S}\right\|_2\nonumber\\
&\leq \sum_{j=0}^p\beta^{p-j}\|(x^j-x_S+\hat{r}^j)_{\Omega_k\setminus S}\|_2.
\end{align}
Due to $(\Omega_k\setminus S)\cup S=\Omega_k\cup S$ and $(\Omega_k\setminus S)\cap S=\emptyset$, one has
 \begin{align*}
 \|(x^j     - x_S   + \hat{r}^j)_{S}\|_2   &   +\|(x^j-x_S+ \hat{r}^j)_{\Omega_k\setminus S}\|_2\nonumber\\
&\leq\sqrt{2}\sqrt{\|(x^j-x_S+ \hat{r}^j)_{S}\|_2^2+\|(x^j-x_S+ \hat{r}^j)_{\Omega_k\setminus S}\|_2^2} \nonumber \\
& = \sqrt{2}\|(x^j-x_S+ \hat{r}^j)_{\Omega_k\cup S }\|_2,
\end{align*}
which together with \eqref{up-x-bd1}-\eqref{up-x-bd3-1} implies that
\begin{align}\label{up-x-bd4}
\|(u^p-x_S)_S\|_2
\leq &\sqrt{2}\sum_{j=0}^p\beta^{p-j}\|(x^j-x_S+ \hat{r}^j)_{\Omega_k\cup S }\|_2\nonumber\\
& ~~~ +\sum_{j=0}^{p-1}\beta^{p-j}\|x^j-x_S\|_2+\|(r^p)_{\Omega_k\setminus \Omega_q}\|_2.
\end{align}
Note that $(\Omega_k\setminus \Omega_q)\cup\Omega_q=\Omega_k$ and $(\Omega_k\setminus \Omega_q)\cap\Omega_q=\emptyset. $ We have
\begin{align*}
\|(r^p)_{\Omega_k\setminus \Omega_q}\|_2^2+\|(r^p)_{\Omega_q}\|_2^2=\|(r^p)_{\Omega_k}\|_2^2.
\end{align*}
By Lemma \ref{lem-rp_Omg_qk}, we have
\begin{align}\label{up-x-bd5}
\|(r^p)_{\Omega_k\setminus \Omega_q}\|_2
\leq \sqrt{1-[g(\gamma)]^2}\|(r^p)_{\Omega_k}\|_2,
\end{align}
in which the term $\|(r^p)_{\Omega_k}\|_2$ can be bounded as
\begin{align}\label{up-x-bd6}
\|(r^p)_{\Omega_k}\|_2
&=\left\|\sum_{j=0}^p\beta^{p-j}(x^j-x_S+ \hat{r}^j)_{\Omega_k }-\sum_{j=0}^p\beta^{p-j}(x^j-x_S)_{\Omega_k }\right\|_2\nonumber\\
&\leq\sum_{j=0}^p\beta^{p-j}\|(x^j-x_S+ \hat{r}^j)_{\Omega_k }\|_2+\sum_{j=0}^p\beta^{p-j}\|x^j-x_S\|_2\nonumber\\
&\leq\sum_{j=0}^p\beta^{p-j}\|(x^j-x_S+ \hat{r}^j)_{\Omega_k \cup S }\|_2+\sum_{j=0}^p\beta^{p-j}\|x^j-x_S\|_2.
\end{align}
 From \eqref{up-x-bd4}-\eqref{up-x-bd6}, it is easy to obtain that
\begin{align}\label{up-x-bd7}
\|(u^p-x_S)_S\|_2
\leq &\left(\sqrt{2}+ \sqrt{1-[g(\gamma)]^2}\right)\sum_{j=0}^p\beta^{p-j}\|(x^j-x_S+ \hat{r}^j)_{\Omega_k\cup S }\|_2\nonumber\\
&+\sum_{j=0}^{p-1}\beta^{p-j}\|x^j-x_S\|_2+ \sqrt{1-[g(\gamma)]^2}\sum_{j=0}^p\beta^{p-j}\|x^j-x_S\|_2.
\end{align}
Since $\hat{r}^j=A^ T(y-Ax^j)$ and $y=Ax_S+\nu'$, we have
\begin{align}\label{xxpx}
x^j-x_S+ \hat{r}^j=(I-A^TA)(x^j-x_S)+ A^T\nu',~~ j=0,\ldots,p.
\end{align}
By using \eqref{xxpx} and triangle inequality, we see that for each $ j=0,\ldots,p,$
\begin{align}\label{up-x-bd8}
\|(x^j-x_S+ \hat{r}^j)_{\Omega_k\cup S}\|_2
\leq&\|[(I-A^TA)(x^j-x_S)]_{\Omega_k\cup S}\|_2+\|(A^T\nu')_{\Omega_k\cup S}\|_2\nonumber\\
\leq &\delta_{3k}\|x^j-x_S\|_2+\sqrt{1+\delta_{2k}}\|\nu'\|_2,
\end{align}
where the last inequality follows from Lemma \ref{lem-basic-ineq} with $|\textrm{supp}(x^j-x_S)\cup(\Omega_k \cup S)|\leq 3k$ and $|\Omega_k \cup S|\leq 2k$.
Inserting \eqref{up-x-bd8} into \eqref{up-x-bd7} yields
\begin{align}
\|(u^p-x_S)_S\|_2
\leq &\left(\sqrt{2}+ \sqrt{1-[g(\gamma)]^2}\right)\left(\delta_{3k}\sum_{j=0}^p\beta^{p-j}\|x^j-x_S\|_2+\frac{\sqrt{1+\delta_{2k}}}{1-\beta}\|\nu'\|_2\right)\nonumber\\
&+\sum_{j=0}^{p-1}\beta^{p-j}\|x^j-x_S\|_2+ \sqrt{1-[g(\gamma)]^2}\sum_{j=0}^p\beta^{p-j}\|x^j-x_S\|_2,\nonumber
\end{align}
which is \eqref{bound-up-xS} by setting $Q_{p-1}, Q_p, C_1$ and $ C_2$ as \eqref{Qp} and \eqref{C12}.
\end{proof}

 We need one more technical result before showing the main result.
\begin{Lem}\label{property-Gt}
For any given $\gamma\in (0,1]$, let $g(\gamma)$ be given by (\ref{g-gam}). Then the function
\begin{align}\label{Gt}
G(t)=\frac{1}{1-t}\left[\sqrt{2}t+t\sqrt{\frac{5+t}{1+t}}+\sqrt{1-(g(\gamma))^2}(1+t)\right]-1,~t\in [0,1)
\end{align}
is strictly increasing and has a unique root, denoted by $\delta(\gamma),$ in $(0,1).$
\end{Lem}

\begin{proof} Note that $ \frac{t}{1-t}$ and $ \frac{1+t}{1-t}$ are strictly increasing in  $[0,1)$ and that
\begin{align}\label{Gt-part-inc}
\frac{1}{1-t}\cdot\frac{1}{\sqrt{1+t}}=\frac{1}{\sqrt{1-t}}\cdot\frac{1}{\sqrt{1-t^2}}
\end{align}
 is  strictly increasing in  $[0,1),$ so is  $\frac{t}{1-t} \sqrt{\frac{5+t}{1+t}}.$  Thus
the function $G(t)$  in \eqref{Gt}  is strictly increasing in  $[0,1)$ for any given $\gamma\in(0,1]$. For a fixed $\gamma\in(0,1]$, $G(t)$ is continuous function over $[0,1)$ satisfying that $G(0)=\sqrt{1-[g(\gamma)]^2}-1<0$ and $\lim_{t\rightarrow 1^-}G(t)=+\infty$. Thus,  $G(t)=0$ has a unique root in $(0,1),$  denoted by $ \delta(\gamma).$
\end{proof}

\begin{Rem} \emph{Compared with the analysis of related algorithms,  the main difficulty in the analysis of this paper (due to appearance of generalized means functions) is to establish some new fundamental technical results that are used to show the main result.  Lemmas \ref{lem-rp_Omg_qk}  and  \ref{bound_up-x} are among such technical results. In Lemma \ref{lem-rp_Omg_qk}, we establish the relation of $\|(r^p)_{\Omega_q}\|_{2}$ ($q\leq k$) and $\|(r^p)_{\Omega_k}\|_2$, which is rooted on the convexity and  monotonicity of the generalized mean function.  Furthermore,  with the aid of Lemma \ref{lem-rp_Omg_qk}, we establish in Lemma \ref{bound_up-x} the upper bound of  $\|(u^p-x_S)_S\|_2$ in terms of the linear combination of $\|x^j-x_S\|_2,~j=0,\ldots,p.$ This bound is essential to establish the solution error bound of DTAM which are summarized in the theorem below. Moreover, as a by-product of our analysis (see Corollary \ref{Cor-PG} for details), we can also establish the error bound of PGROTP for the case  $\bar{q}=k$, which has not obtained based on the analysis in \cite{MZKS22}.}
\end{Rem}

The main result for DTAM is stated as follows.

\begin{Thm}\label{main-thm}
 Let $x \in \mathbb{R}^n$ be the solution to the linear inverse problem $y = Ax+\nu$ where $\nu$ is a noise vector.
For any given $\gamma\in (0,1]$, suppose that the RIC, $\delta_{3k},$ of  matrix A and the forgetting factor $\beta$ satisfy that
 \begin{align}\label{beta-range}
\delta_{3k}<\delta(\gamma),~~ 0\leq\beta<\frac{2\tilde{\varrho}}{\delta_{2k}+\sqrt{(\delta_{2k})^2+4\tilde{\varrho}(1-\delta_{2k})}}-\tilde{\varrho},
 \end{align}
 where $\delta(\gamma)\in(0,1)$ is given in Lemma \ref{property-Gt} and
 \begin{align}\label{rho-tilde}
\tilde{\varrho}:=&\frac{1}{1-\delta_{2k}}\left(C_1+\delta_{3k}\sqrt{\frac{5+\delta_{2k}}{1+\delta_{2k}}}\right) <1
\end{align}
with $C_1$ is given by \eqref{C12}. Then
 the sequence $\{x^{p}\}$ generated by DTAM satisfies
\begin{align}\label{er-bd}
\|x^{p}-x_S\|_2\leq \varrho^p\|x^{0}-x_S\|_2+\frac{C_\beta}{1-\varrho}\|\nu'\|_2,
\end{align}
where $S=\mathcal{L}_k(x)$, $\nu'=Ax_{\overline{S}}+\nu =y-A x_S,$ and
\begin{align}\label{rho-C}
\varrho &:=\tilde{\varrho}+\beta+\frac{\beta}{(1-\delta_{2k})(\tilde{\varrho}+\beta)}<1,\\
C_\beta &:=\frac{1}{1-\delta_{2k}}\left[\frac{C_2+\sqrt{5+\delta_{2k}}}{1-\beta}+\frac{2}{\sqrt{1+\delta_{2k}}}+\sqrt{1+\delta_{k}}\right],\nonumber
\end{align}
in which $C_2$ is given by \eqref{C12}.
\end{Thm}

\begin{proof} The proof  is partitioned into the three parts.

{\it Part I. We first show that under the condition of the theorem, the constants $ \tilde{\varrho},\varrho$  in \eqref{rho-tilde} and \eqref{rho-C} are smaller than 1, and that the range for $\beta$ in \eqref{beta-range} is well-defined. }

In fact, since the function in \eqref{Gt-part-inc} is strictly increasing in  $[0,1)$, from the fact $\delta_{2k}\leq\delta_{3k}<\delta(\gamma)<1$, we immediately see that
\begin{align}\label{parti324}
\frac{1}{1-\delta_{2k}}\cdot\frac{1}{\sqrt{1+\delta_{2k}}}\leq \frac{1}{1-\delta_{3k}}\cdot\frac{1}{\sqrt{1+\delta_{3k}}}.
\end{align}
It follows from \eqref{C12}, \eqref{rho-tilde} and Lemma \ref{property-Gt} that
$$
\tilde{\varrho}\leq G(\delta_{3k})+1<G(\delta(\gamma))+1=1,
$$
where the second inequality follows from the fact that $G(t)$  is strictly increasing in  $[0,1)$ and the equality follows from $G(\delta(\gamma))=0$. Since $\tilde{\varrho}<1$ and $\delta_{2k}<1$, we have
$$\frac{2\tilde{\varrho}}{\delta_{2k}+\sqrt{(\delta_{2k})^2+4\tilde{\varrho}(1-\delta_{2k})}}
>\frac{2\tilde{\varrho}}{\delta_{2k}+\sqrt{(\delta_{2k})^2+4(1-\delta_{2k})}}=\frac{2\tilde{\varrho}}{\delta_{2k}+\sqrt{(2-\delta_{2k})^2}}=\tilde{\varrho}.$$
Thus  the range for $\beta$ in \eqref{beta-range} is well-defined. By setting $\zeta:=\frac{1}{1-\delta_{2k}}(>1),$ the second inequality
in \eqref{beta-range} can be written as
 $$
0\leq\beta<\left(\sqrt{(2\tilde{\varrho}+\zeta-1)^2+4\tilde{\varrho}(1-\tilde{\varrho})}-(2\tilde{\varrho}+\zeta-1)\right)/2.
$$
This implies that
$$
\beta^2+(2\tilde{\varrho}+\zeta-1)\beta-\tilde{\varrho}(1-\tilde{\varrho})<0,
$$
which is equivalent to $\varrho<1,$ as sated in  \eqref{rho-C}.

{\it Part II. We now estimate the term $\|x^{p+1}-x_S\|_2$ in terms of $\|(x_S-u^p)_{V^p\setminus S}\|_2$  and  $\|(x_S-u^p)_S\|_2.$ The upper bound for this term is key to establishing the desired error bound in \eqref{er-bd}.}

 {\bf Case 1. $|V^p|>k$.} In this case, $S^{p+1}=\mathcal{L}_k(u^p \circ w^p)\subset V^p$ and $u^p$, $w^p$ are given by \eqref{SDir} and \eqref{alg-DSROTP-2} respectively. Set $u^*=x^{p+1}, ~\Omega=S^{p+1}$ and  $S=\mathcal{L}_k(x)$  in Lemma \ref{lem-pursuit}, we get
\begin{align}\label{er-bd-projection}
\|x^{p+1}-x_S\|_2\leq \frac{1}{\sqrt{1-(\delta_{2k})^2}}\| (x_S)_{\overline{S^{p+1}}}\|_2+\frac{ \sqrt{1+\delta_{k}}}{1-\delta_{2k}}\|\nu' \|_2.
\end{align}
Since $\textrm{supp}(\mathcal{H}_k(u^p \circ w^p))\subseteq S^{p+1}$, it follows that
\begin{align}\label{xp1-Hkup}
\|x^{p+1}-x_S\|_2
\leq&\frac{1}{\sqrt{1-(\delta_{2k})^2}}\|( \mathcal{H}_k(u^p \circ w^p)-x_S)_{\overline{S^{p+1}}}\|_2+\frac{ \sqrt{1+\delta_{k}}}{1-\delta_{2k}}\|\nu' \|_2\nonumber\\
\leq&\frac{1}{\sqrt{1-(\delta_{2k})^2}}\|\mathcal{H}_k(u^p \circ w^p)-x_S\|_2+\frac{ \sqrt{1+\delta_{k}}}{1-\delta_{2k}}\|\nu' \|_2.
\end{align}
 Now, we can bound the term  $\|x_S-\mathcal{H}_k(u^p \circ w^p)\|_2$ by using $\|(x_S-u^p)_{V^p\setminus S}\|_2$, $\|(x_S-u^p)_S\|_2$ and  $\|\nu' \|_2$. By  Lemma \ref{lem-hard-thres}, we have
\begin{align}\label{x-Hk-uw}
\|x_S-\mathcal{H}_k(u^p \circ w^p)\|_2\leq \|(x_S-u^p \circ w^p)_{S\cup S^{p+1}}\|_2+\|(x_S-u^p \circ w^p)_{ S^{p+1}\setminus S}\|_2.
\end{align}
Note that $y=Ax_S+\nu'$ with $\nu'=Ax_{\overline{S}}+\nu$ and $\textrm{supp}(u^p)\subseteq V^p. $ Using the triangle inequality leads to
\begin{align*}
 \|y  & -A(u^p \circ w^p)\|_2\nonumber\\
&=\|A(x_S-u^p \circ w^p)_{S\cup S^{p+1}}+A(x_S-u^p \circ w^p)_{V^p\setminus (S\cup S^{p+1})}+\nu'\|_2\nonumber\\
&\geq\|A(x_S-u^p \circ w^p)_{S\cup S^{p+1}}\|_2-\|A(x_S-u^p \circ w^p)_{V^p\setminus (S\cup S^{p+1})}\|_2-\|\nu'\|_2\nonumber\\
&\geq\sqrt{1-\delta_{2k}}\|(x_S-u^p \circ w^p)_{S\cup S^{p+1}}\|_2 \nonumber \\
 & ~~~~ -\sqrt{1+\delta_{2k}}\|(x_S-u^p \circ w^p)_{V^p\setminus (S\cup S^{p+1})}\|_2-\|\nu'\|_2,
\end{align*}
where the last inequality follows  from \eqref{def-RIC-1} with $|S\cup S^{p+1}|\leq 2k$ and $|V^p\setminus (S\cup S^{p+1})|\leq 2k$. Thus
\begin{align}\label{xuw-bd1}
\|(x_S-u^p \circ w^p)_{S\cup S^{p+1}}\|_2\leq &\sqrt{\frac{1+\delta_{2k}}{1-\delta_{2k}}}\|(x_S-u^p \circ w^p)_{V^p\setminus (S\cup S^{p+1})}\|_2\nonumber\\
&+\frac{1}{\sqrt{1-\delta_{2k}}}(\|y-A(u^p \circ w^p)\|_2+\|\nu'\|_2).
\end{align}
Due to $(x_S)_{V^p\setminus (S\cup S^{p+1})}=(x_S)_{ S^{p+1}\setminus S}=0$ and   $0\leq w^p \leq  {\bf e}$, we obtain
\begin{align*}
\|(x_S-u^p \circ w^p)_{V^p\setminus (S\cup S^{p+1})}\|_2 &=\|[(x_S-u^p) \circ w^p]_{V^p\setminus (S\cup S^{p+1})}\|_2 \\ & \leq \|(x_S-u^p)_{V^p\setminus (S\cup S^{p+1})}\|_2
\end{align*}
and
\begin{align*}
\|(x_S-u^p \circ w^p)_{ S^{p+1}\setminus S}\|_2&=\|[(x_S-u^p) \circ w^p]_{ S^{p+1}\setminus S}\|_2
\leq \|(x_S-u^p)_{ S^{p+1}\setminus S}\|_2.
\end{align*}
Combining  the two inequalities above with \eqref{x-Hk-uw} and \eqref{xuw-bd1} yields
\begin{align}\label{x-Hk-uw1}
\|x_S-\mathcal{H}_k(u^p \circ w^p)\|_2\leq &\sqrt{\frac{1+\delta_{2k}}{1-\delta_{2k}}}\|(x_S-u^p)_{V^p\setminus (S\cup S^{p+1})}\|_2+\|(x_S-u^p)_{ S^{p+1}\setminus S}\|_2\nonumber\\
&+\frac{1}{\sqrt{1-\delta_{2k}}}(\|y-A(u^p \circ w^p)\|_2+\|\nu'\|_2).
\end{align}
We now further estimate the term $\|y-A(u^p \circ w^p)\|_2. $   Note that $S^{p+1}\subset V^p$ and  $|S|=|S^{p+1}|=k.$ Let $\hat{w}\in\{0,1\}^n$ be a $k$-sparse vector in the feasible set of the problem \eqref{alg-DSROTP-2} such that $\hat{w}_i=1 $ for all $i\in V^p\cap S  $ and $\hat{w}_j=0$ for all $j\in V^p\setminus (S^{p+1}\cup S). $ Then
\begin{align}\label{yAuwp-bd}
\|y-A(u^p \circ w^p)\|_2\leq&\|y-A(u^p \circ \hat{w})\|_2
\leq\|A(x_S-u^p \circ \hat{w})\|_2+\|\nu'\|_2\nonumber\\
\leq&\sqrt{1+\delta_{2k}}\|x_S-u^p \circ \hat{w}\|_2+\|\nu'\|_2,
\end{align}
where  the first inequality is due to $w^p$ being the optimal solution to \eqref{alg-DSROTP-2}, the second inequality follows from $y= Ax_S+\nu'$, and the  third follows from \eqref{def-RIC-1} since   $x_S-u^p \circ \hat{w}$ is  $(2k)$-sparse. Note that  $$(u^p \circ \hat{w})_{V^p\cap S}=(u^p)_{V^p\cap S}, ~\hat{w}_{V^p\setminus S}=\hat{w}_{S^{p+1}\setminus S} $$ and $\textrm{supp}(u^p)\subseteq V^p$ and $(x_S)_{S^{p+1}\setminus S}=0.$  We deduce that
\begin{align}\label{xup_omhat}
\|x_S-u^p \circ \hat{w}\|_2
&=\|x_S-(u^p)_{V^p\cap S}-u^p \circ \hat{w}_{V^p\setminus S}\|_2\nonumber\\
&=\|x_S-(u^p)_{S}+(x_S-u^p) \circ \hat{w}_{S^{p+1}\setminus S}\|_2\nonumber\\
&\leq\|(x_S-u^p)_{S}\|_2+\|(x_S-u^p)_{S^{p+1}\setminus S}\|_2.
\end{align}
Inserting \eqref{xup_omhat} into \eqref{yAuwp-bd} leads to
\begin{align}\label{yAuwp-bd-1}
\|y-A(u^p \circ w^p)\|_2
\leq&\sqrt{1+\delta_{2k}}(\|(x_S-u^p)_{S}\|_2+\|(x_S-u^p)_{S^{p+1}\setminus S}\|_2)+\|\nu'\|_2.
\end{align}
Merging \eqref{x-Hk-uw1} with \eqref{yAuwp-bd-1} leads to
\begin{align}\label{x-Hk-uw1-0}
&\|x_S-\mathcal{H}_k(u^p \circ w^p)\|_2\nonumber\\
&\leq\sqrt{\frac{1+\delta_{2k}}{1-\delta_{2k}}}\|(x_S-u^p)_{V^p\setminus (S\cup S^{p+1})}\|_2+\left(\sqrt{\frac{1+\delta_{2k}}{1-\delta_{2k}}}+1\right)
\|(x_S-u^p)_{ S^{p+1}\setminus S}\|_2\nonumber\\
&~~~ +\sqrt{\frac{1+\delta_{2k}}{1-\delta_{2k}}}\|(x_S-u^p)_S\|_2+\frac{2}{\sqrt{1-\delta_{2k}}}\|\nu'\|_2.
\end{align}
Denote by $$\Delta_1:=\|(x_S-u^p)_{V^p\setminus (S\cup S^{p+1})}\|_2, ~ \Delta_2:=\|(x_S-u^p)_{S^{p+1}\setminus S}\|_2$$ and $\Delta:=\|(x_S-u^p)_{V^p\setminus S}\|_2$. Note that $V^p\setminus S=[V^p\setminus (S\cup S^{p+1})]\cup (S^{p+1}\setminus S)$ and  $[V^p\setminus (S\cup S^{p+1})]\cap (S^{p+1}\setminus S)=\emptyset.$ We have $\Delta_1^2+\Delta_2^2=\Delta^2$. Hence, for any given number $a,b>0,$  we have
\begin{align}\label{inq_Delta}
a\Delta_1+b\Delta_2\leq \sqrt{b^2+a^2}\sqrt{\Delta_1^2+\Delta_2^2}=\sqrt{b^2+a^2}\Delta.
\end{align}
In particular, if $b=a+1$, then \eqref{inq_Delta} becomes
\begin{align}\label{inq_Delta2}
a\Delta_1+(a+1)\Delta_2\leq \sqrt{(a+1)^2+a^2}\Delta\leq \sqrt{3a^2+2}\Delta.
\end{align}
By setting $a=\sqrt{\frac{1+\delta_{2k}}{1-\delta_{2k}}}$ in \eqref{inq_Delta2}, we see that \eqref{x-Hk-uw1-0} becomes
\begin{align*}%\label{x-Hk-uw1-1}
&\|x_S-\mathcal{H}_k(u^p \circ w^p)\|_2\nonumber\\
&\leq\sqrt{\frac{5+\delta_{2k}}{1-\delta_{2k}}}\|(x_S-u^p)_{V^p\setminus S}\|_2+\sqrt{\frac{1+\delta_{2k}}{1-\delta_{2k}}}\|(x_S-u^p)_S\|_2+\frac{2}{\sqrt{1-\delta_{2k}}}\|\nu'\|_2.
\end{align*}
Substituting this into \eqref{xp1-Hkup} leads to
\begin{align}\label{xp1-x-1}
\|x^{p+1}-x_S\|_2
\leq&\frac{1}{1-\delta_{2k}}\left(\sqrt{\frac{5+\delta_{2k}}{1+\delta_{2k}}}\|(x_S-u^p)_{V^p\setminus S}\|_2+\|(x_S-u^p)_S\|_2\right)\nonumber\\
&+\frac{1}{1-\delta_{2k}}\left(\frac{2}{\sqrt{1+\delta_{2k}}}+\sqrt{1+\delta_{k}}\right)\|\nu'\|_2.
\end{align}
 {\bf Case 2. $|V^p|\leq k$.}   In this case, $\textrm{supp}(u^p)\subseteq V^p=S^{p+1}$, and hence $(u^p)_{\overline{S^{p+1}}}=0$. Thus,
\begin{align}\label{up-xSp}
\| (x_S)_{\overline{S^{p+1}}}\|_2=\| (x_S)_{S\setminus{S^{p+1}}}\|_2=\| (u^p-x_S)_{S\setminus{S^{p+1}}}\|_2\leq \| (u^p-x_S)_S\|_2.
\end{align}
Substituting \eqref{up-xSp} into \eqref{er-bd-projection}, we have
\begin{align*}
\|x^{p+1}-x_S\|_2
\leq &\frac{1}{\sqrt{1-(\delta_{2k})^2}}\| (u^p-x_S)_{S}\|_2+\frac{ \sqrt{1+\delta_{k}}}{1-\delta_{2k}}\|\nu' \|_2.
\end{align*}
Compared with \eqref{xp1-x-1}, the inequality \eqref{xp1-x-1} remains valid for the case  $|V^p|\leq k$.

{\it Part III. We now further establish the error bound \eqref{er-bd} via the mathematical induction.}

 (i) Clearly, \eqref{er-bd}  holds for $p=0$.

 (ii) For $p\geq 1,$  assume that
\begin{align}\label{xj-er-bd}
\|x^{j}-x_S\|_2\leq \varrho^j\|x^{0}-x_S\|_2+\frac{C_\beta}{1-\varrho}\|\nu'\|_2
\end{align}
holds for $j=0,\ldots, p.$ We need to show that \eqref{xj-er-bd} holds for $j=p+1$. Similar to \eqref{bound-up-xS}, the upper bound of $\|(x_S-u^p)_{V^p\setminus S}\|_2$ can be determined in terms of $Q_p$ and $\|\nu' \|_2$.
By the definition of $ u^p$ in \eqref{alg-DSROTP-1} and noting that  $V^p=\textrm{supp}(x^p)\cup\Omega_q$ and $r^p=\sum_{j=0}^p\beta^{p-j}\hat{r}^j, $  we obtain
\begin{align}\label{up-x-CS-0-0}
\|(u^p-x_S)_{V^p\setminus S}\|_2
=\left\|(x^p+ (\hat{r}^p)_{\Omega_q}-x_S)_{V^p\setminus S}+\sum_{j=0}^{p-1}\beta^{p-j}(\hat{r}^j)_{\Omega_q\setminus S}\right\|_2.
\end{align}
From \eqref{relation-r-x}, we see that $(\hat{r}^p)_{\Omega_q}=(\hat{r}^p)_{V^p}$ and $(\hat{r}^j)_{\Omega_q\setminus S}=(x^j-x_S+\hat{r}^j)_{\textrm{supp}(\hat{r}^j)\cap \Omega_q\setminus S}$. It follows from \eqref{up-x-CS-0-0} that
\begin{align}\label{up-x-CS-0}
\|(u^p-x_S)_{V^p\setminus S}\|_2
=& \left\|(x^p+ \hat{r}^p-x_S)_{V^p\setminus S}+\sum_{j=0}^{p-1}\beta^{p-j}(x^j-x_S+\hat{r}^j)_{\textrm{supp}(\hat{r}^j)\cap \Omega_q\setminus S}\right\|_2\nonumber\\
\leq & \left\|(x^p+ \hat{r}^p-x_S)_{V^p\setminus S}\right\|_2+\sum_{j=0}^{p-1}\beta^{p-j}\left\|(x^j-x_S+\hat{r}^j)_{ \Omega_q\setminus S}\right\|_2.
\end{align}
 Similar to \eqref{up-x-bd8}, replacing the index set $\Omega_k \cup S$ with $V^p\setminus S$ and $\Omega_q\setminus S$ respectively, we obtain that
$$
\|(x^p-x_S+\hat{r}^p)_{ V^p\setminus S}\|_2\leq  \delta_{3k}\|x^p-x_S\|_2+\sqrt{1+\delta_{2k}}\|\nu'\|_2
$$
and
$$
\|(x^j-x_S+\hat{r}^j)_{\Omega_q\setminus S}\|_2\leq  \delta_{3k}\|x^j-x_S\|_2+\sqrt{1+\delta_{2k}}\|\nu'\|_2,
~~j=0,\ldots,p-1.
$$
Combining the two inequalities above with \eqref{up-x-CS-0} leads to
\begin{align}\label{up-x-CS}
\|(u^p-x_S)_{V^p\setminus S}\|_2
\leq \delta_{3k} Q_p+\frac{\sqrt{1+\delta_{2k}}}{1-\beta}\|\nu'\|_2,
\end{align}
where $Q_p$ is given by \eqref{Qp}.
With the aid of  \eqref{bound-up-xS}  and \eqref{up-x-CS}, the inequality \eqref{xp1-x-1} can be written further as
\begin{align}\label{xp1-x-0}
\|x^{p+1}-x_S\|_2
\leq&\frac{1}{1-\delta_{2k}}\left[\left(C_1+\delta_{3k}\sqrt{\frac{5+\delta_{2k}}{1+\delta_{2k}}} \right)Q_p+\beta Q_{p-1}\right]+C_\beta\|\nu'\|_2\nonumber\\
=&\tilde{\varrho} Q_p+ \frac{\beta }{1-\delta_{2k}} Q_{p-1}+C_\beta\|\nu'\|_2,
\end{align}
where $\tilde{\varrho}$, $C_\beta$ are given by \eqref{rho-tilde} and \eqref{rho-C}, respectively.
Inserting \eqref{xj-er-bd}  into \eqref{Qp} leads to
\begin{align}\label{Qp-varrho}
Q_i&\leq \sum_{j=0}^i\beta^{i-j}\varrho^j\|x^{0}-x_S\|_2+\frac{C_\beta(1-\beta^{i+1})}{(1-\varrho)(1-\beta)}\|\nu'\|_2\nonumber\\
&\leq\varrho^i\frac{1-(\beta/\varrho)^{i+1}}{1-\beta/\varrho}\|x^{0}-x_S\|_2+\frac{C_\beta}{(1-\varrho)(1-\beta)}\|\nu'\|_2\nonumber\\
&\leq\frac{\varrho^{i+1}}{\varrho-\beta}\|x^{0}-x_S\|_2+\frac{C_\beta}{(1-\varrho)(1-\beta)}\|\nu'\|_2
\end{align}
for $i=p-1,p$, in which the condition $\beta<\varrho<1$ is used  and $\varrho$  is given by \eqref{rho-C}.
Substituting \eqref{Qp-varrho} into \eqref{xp1-x-0}  yields
\begin{align}\label{xp1-x}
\|x^{p+1}-x_S\|_2
\leq&\tilde{\varrho} \left(\frac{\varrho^{p+1}}{\varrho-\beta}\|x^{0}-x_S\|_2+\frac{C_\beta}{(1-\varrho)(1-\beta)}\|\nu'\|_2\right)\nonumber\\
&+\frac{\beta }{1-\delta_{2k}} \left(\frac{\varrho^{p}}{\varrho-\beta}\|x^{0}-x_S\|_2+\frac{C_\beta}{(1-\varrho)(1-\beta)}\|\nu'\|_2\right)+C_\beta\|\nu'\|_2\nonumber\\
\leq&\left(\varrho\tilde{\varrho} +\frac{\beta }{1-\delta_{2k}} \right)\frac{\varrho^{p}}{\varrho-\beta}\|x^{0}-x_S\|_2\nonumber\\
&+\left[\left(\tilde{\varrho} +\frac{\beta }{1-\delta_{2k}} \right)\frac{1}{1-\beta}+1-\varrho\right]\frac{C_\beta}{1-\varrho}\|\nu'\|_2.
\end{align}
To simplify \eqref{xp1-x}, we need to  estimate the coefficients of $\|x^{0}-x_S\|_2$ and $\|\nu'\|_2$.
Using the definition of $\varrho$  in \eqref{rho-C}, we have
\begin{align}
\varrho
\geq&\tilde{\varrho}+\beta+\frac{2\beta}{(1-\delta_{2k})\left(\tilde{\varrho}+\beta+\sqrt{(\tilde{\varrho}+\beta)^2+\frac{4\beta}{1-\delta_{2k}}}\right)}\nonumber\\
=&\tilde{\varrho}+\beta+\frac{\sqrt{(\tilde{\varrho}+\beta)^2+\frac{4\beta}{1-\delta_{2k}}}-\tilde{\varrho}-\beta}{2}\nonumber\\
=&\frac{\tilde{\varrho}+\beta+\sqrt{(\tilde{\varrho}+\beta)^2+\frac{4\beta}{1-\delta_{2k}}}}{2}\nonumber,
\end{align}
which implies that
$
\varrho^2-\varrho(\tilde{\varrho}+\beta)-\frac{\beta}{1-\delta_{2k}}\geq 0.
$
This is equivalent to
\begin{align}\label{coefficient-1}
\frac{1}{\varrho-\beta}\left(\varrho\tilde{\varrho}+\frac{\beta}{1-\delta_{2k}}\right)\leq\varrho.
\end{align}
It follows from $\varrho<1$ in \eqref{rho-C} that
\begin{align}\label{coefficient-2}
\left(\tilde{\varrho}+\frac{\beta }{1-\delta_{2k}} \right)\frac{1}{1-\beta}+1-\varrho
\leq\frac{\varrho-\beta}{1-\beta}+1-\varrho\leq 1.
\end{align}
By \eqref{coefficient-1} and \eqref{coefficient-2}, we obtain from \eqref{xp1-x} the  inequality
$$
\|x^{p+1}-x_S\|_2
\leq\varrho^{p+1}\|x^{0}-x_S\|_2+\frac{C_\beta}{1-\varrho}\|\nu'\|_2.
$$
Thus \eqref{xj-er-bd}   holds for $j=p+1$. We  conclude that
\eqref{er-bd} holds for all nonnegative integers $p$.
\end{proof}

\begin{Rem} \label{Rem-Zhao}  \emph{The main result in this section discloses the theoretical (guaranteed) performance  of the DTAM under the condition (\ref{beta-range}). This condition also indicates that the choice of general mean functions may influence the performance of the algorithm. From Theorem \ref{main-thm},  one can see that the selection of the generalized mean function would determine the value of $g(\gamma)$  and thus directly affect the constants $ C_1, C_2, \rho $,$\varrho$, $\tilde{\varrho}$ and $ \delta(\gamma). $ This influences the error bound and condition (\ref{beta-range})  itself.  More specifically,  let us assume the target data $x$ being $k$-sparse and $\nu=0$, and thus $\|\nu'\|_2=0$ in (\ref{er-bd}).  From (\ref{er-bd}), we see  that the smaller $\varrho$  is, the faster the convergence speed of the algorithm would be. By simply taking $ \beta=0$, we immediately see that  $\varrho=\tilde{\varrho}$ which is decreasing  with respect to $g(\gamma)$. Thus in theory, one can choose generalized mean functions such that the constant $ \rho $ in error bound  is as small as possible so that the algorithm can converge as quickly as possible.}
\end{Rem}

\begin{Rem}\label{rem-coef-noise}\emph{
From \eqref{C12}  and \eqref{rho-tilde} and Definition \ref{def-RIC},  the constants in \eqref{rho-C} including $\delta_k,~\delta_{2k},~C_2$ and  $\tilde{\varrho}$  only depend on either the matrix $A$ or the parameter $\gamma$ together with the general mean function being used. These constants are independent of the noise level $ \|\nu'\|_2.$  From \eqref{rho-C},  $\varrho$ and $C_\beta$  are strictly increasing  with respect to $\beta$ for given $\gamma\in (0,1].$ This indicates that the coefficient $\frac{C_\beta}{1-\varrho}$ in \eqref{er-bd} is strictly increasing  with respect to $\beta$. To control this coefficient, we may use a relatively small $ \beta$ when the noise level is relatively high; otherwise, we may use a relatively large $\beta$ when the noise level is low.}
\end{Rem}

\begin{Rem}\label{rem-RIC-ND-matrix}\emph{
From \cite[Proposition 6.2]{FR13}, we see that the ($3k$)-th order RIC of the matrix $A$ satisfies $\delta_{3k}\leq (3k-1)\mu$  if $A$ has $\ell_2$-normalized columns, where $\mu$ is the coherence of $A$. Moreover,  the coherence of the normalized gaussian matrix $A$ satisfies $$\mu\leq\frac{\sqrt{15\log n}}{\sqrt{m}-\sqrt{12\log n}}$$ with probability exceeding $1-11/n$ if $60\log n\leq m \leq (n-1)/(4\log n)$   (see, e.g., \cite[Theorem 2]{MBC11}). Thus the ($3k$)-th order RIC of the normalized gaussian matrix $A$ satisfies $\delta_{3k}<\delta(\gamma)$ in \eqref{beta-range} with probability exceeding $1-11/n$ provided that $m_{min}< m \leq (n-1)/(4\log n),$ where $$m_{min}=\max\left\{20,\left(\frac{\sqrt{5}(3k-1)}{\delta(\gamma)}+2\right)^2\right\}\cdot 3\log n.$$
Moreover, the inequality $m_{min}< (n-1)/(4\log n)$ is ensured for given $\delta(\gamma)\in (0,1)$ provided that  $k\ll n$ and $n$ is large enough. \textbf{Based on such observation, we choose to use the gaussian random matrix as the measurement matrix to assess the numerical performance of the algorithm in Section \ref{Num-exp}. }}
\end{Rem}

%======================================================================

 In \cite{MZKS22}, the PGROTP algorithm was  analyzed in the case $\bar{q}\geq 2k.$ The convergence of the algorithm for the case $k\leq \bar{q}< 2k$ was not yet obtained. As a byproduct of the analysis of DTAM in this paper, we can also establish an error bound for PGROTP with $\bar{q}=k.$ This result for PGROTP can be seen as a special case to the main result above.

\begin{Cor}\label{Cor-PG}
Let  $\delta^*$ be the unique root to the univariate function on (0,1)
\begin{align}\label{PG-Gt}
\hat{G}(t):=\frac{\sqrt{2}t}{1-t}\left[1+\frac{1}{\sqrt{1+t}}\right]-1
\end{align}
which is continuous and strictly increasing in  $[0,1)$. Let $x \in \mathbb{R}^n$ be the solution to the system $y = Ax+\nu$ with a noise vector $\nu$.
 If the RIC of the matrix A satisfies $\delta_{3k}<\delta^*\approx 0.272$, then
 the sequence $\{x^{p}\}$ generated by PGROTP with $\bar{q}=k$ obeys
\begin{align}\label{PG-er-bd}
\|x^{p}-x_S\|_2\leq \hat{\varrho}^p\|x^{0}-x_S\|_2+\frac{\hat{C}}{1-\hat{\varrho}}\|\nu'\|_2,
\end{align}
where $\nu'=Ax_{\overline{S}}+\nu$, $S=\mathcal{L}_k(x)$  and
\begin{align*}
\hat{\varrho}&:=\frac{\sqrt{2}\delta_{3k}}{1-\delta_{2k}}\cdot\frac{1+\sqrt{1+\delta_ k}}{\sqrt{1+\delta_{2k}}}<1,\nonumber\\
\hat{C}&:=\frac{1}{1-\delta_{2k}}\left(\sqrt{2}+\frac{2}{\sqrt{1+\delta_{2k}}}+(\sqrt{2}+1)\sqrt{1+\delta_k}\right).
\end{align*}
\end{Cor}

\begin{proof}
Comparing PGROTP with DTAM, we see the following:  (i) The first step of PGROTP with $\bar{q}=k$ is identical to that of DTAM with $\beta=0$ and $q=k.$ Thus for PGROTP, one has  $\|(r^p)_{\Omega_k\setminus \Omega_q}\|_2=0$ in \eqref{up-x-bd5} and $g(\gamma)$ is replaced by 1   in Lemma \ref{lem-rp_Omg_qk}.  (ii) The subproblems \eqref{alg-DSROTP-2}  and \eqref{algorithm-ROTP-2} possess  a common objective function and a constraint  $0\leq w \leq  {\bf e}. $ (iii) The third steps of both algorithms are the identical. Therefore, Lemma \ref{bound_up-x} and  the relations \eqref{er-bd-projection}-\eqref{x-Hk-uw1}, \eqref{inq_Delta} and \eqref{up-x-CS-0}-\eqref{up-x-CS}  in the proof of Theorem \ref{main-thm}  remains valid for PGROTP by simply setting $\beta=0,$  $V^p=\textrm{supp}(x^p)\cup \Omega_k$ and $g(\gamma)= 1$ in previous analysis.

Similar to \eqref{yAuwp-bd}-\eqref{xup_omhat}, we choose a $k$-sparse  vector $\bar{w}\in\{0,1\}^n$ from the feasible set of  \eqref{algorithm-ROTP-2}  such that $\textrm{supp}(\bar{w})=S$, which leads to
\begin{align}\label{PG-xuw}
\|x_S-u^p \circ \bar{w}\|_2=\|x_S-(u^p)_{\textrm{supp}(\bar{w})}\|_2=\|(x_S-u^p)_S\|_2.
\end{align}
 It should be noted that the term $\|(x_S-u^p)_{S^{p+1}\setminus S}\|_2$ is vanished in \eqref{PG-xuw} compared to \eqref{xup_omhat}, due to the choice of  $\bar{w}$ and $\hat{w}$  in the corresponding feasible sets.
 Similar to \eqref{yAuwp-bd}, by using $y=Ax_S+\nu'$ and \eqref{PG-xuw}, we have
\begin{align}\label{PG-yAuwp-bd}
\|y-A(u^p \circ w^p)\|_2\leq&\|y-A(u^p \circ \bar{w})\|_2\nonumber\\
\leq&\|A(x_S-u^p \circ \bar{w})\|_2+\|\nu'\|_2\nonumber\\
\leq&\sqrt{1+\delta_{k}}\|x_S-u^p \circ \bar{w}\|_2+\|\nu'\|_2\nonumber\\
=&\sqrt{1+\delta_{k}}\|(x_S-u^p)_{S}\|_2+\|\nu'\|_2,
\end{align}
where the first inequality is due to $w^p$ being the optimal solution of   \eqref{algorithm-ROTP-2}, and  the  third inequality is ensured by \eqref{def-RIC-1} with the vector $x_S-u^p \circ \bar{w}$ being  $k$-sparse.
Combining \eqref{x-Hk-uw1} with \eqref{PG-yAuwp-bd}, one has
\begin{align}\label{PG-x-Hk-uw}
 \|x_S  &  -\mathcal{H}_k(u^p \circ w^p)\|_2\nonumber\\
&\leq\sqrt{\frac{1+\delta_{2k}}{1-\delta_{2k}}}\|(x_S-u^p)_{V^p\setminus (S\cup S^{p+1})}\|_2+
\|(x_S-u^p)_{ S^{p+1}\setminus S}\|_2\nonumber\\
&~~~ +\sqrt{\frac{1+\delta_{k}}{1-\delta_{2k}}}\|(x_S-u^p)_S\|_2+\frac{2}{\sqrt{1-\delta_{2k}}}\|\nu'\|_2\nonumber\\
&\leq\sqrt{\frac{2}{1-\delta_{2k}}}\|(x_S-u^p)_{V^p\setminus S}\|_2+\sqrt{\frac{1+\delta_{k}}{1-\delta_{2k}}}\|(x_S-u^p)_S\|_2  \nonumber\\
 & ~~+\frac{2}{\sqrt{1-\delta_{2k}}}\|\nu'\|_2,
\end{align}
where the last inequality is from  \eqref{inq_Delta} with $a=\sqrt{\frac{1+\delta_{2k}}{1-\delta_{2k}}}$ and $b=1$.
Setting $\beta=0$ and replacing $g(\gamma)$ by 1 in Lemma \ref{bound_up-x} and  \eqref{up-x-CS}, we obtain
\begin{align}\label{PG-bound-up-xS}
\|(u^p-x_S)_S\|_2 \leq \sqrt{2}\delta_{3k} \|x^p-x_S\|_2+\sqrt{2(1+\delta_{2k})}\|\nu'\|_2
\end{align}
and
\begin{align}\label{PG-up-x-CS}
\|(u^p-x_S)_{V^p\setminus S}\|_2
\leq \delta_{3k} \|x^p-x_S\|_2+\sqrt{1+\delta_{2k}}\|\nu'\|_2.
\end{align}
Substituting \eqref{PG-up-x-CS}  and  \eqref{PG-bound-up-xS} into \eqref{PG-x-Hk-uw} yields
\begin{align*}
\|x_S-\mathcal{H}_k(u^p \circ w^p)\|_2
\leq&\frac{\sqrt{2}\delta_{3k}}{\sqrt{1-\delta_{2k}}}(1+\sqrt{1+\delta_{k}})\|x_S-x^p)\|_2\nonumber\\
&+\sqrt{2}\sqrt{\frac{1+\delta_{2k}}{1-\delta_{2k}}}\left(1+\sqrt{1+\delta_{k}}+\sqrt{\frac{2}{1+\delta_{2k}}}\right)\|\nu'\|_2.
\end{align*}
It follows from \eqref{xp1-Hkup} that
\begin{align*}%\label{PG-xp1-x}
\|x^{p+1}-x_S\|_2\leq\hat{\varrho}\|x^{p}-x_S\|_2+\hat{C}\|\nu'\|_2,
\end{align*} which is the estimation in \eqref{PG-er-bd}. The constants $\hat{\varrho}$ and $ \hat{C}$ are exactly the ones stated in the Corollary.  It is sufficient to show that $\hat{\varrho}<1$.
Note that the function in \eqref{Gt-part-inc} is strictly increasing in  $[0,1)$, it is easy to verify that the function  $\hat{G}(t)$ given in \eqref{PG-Gt} is also strictly increasing in  $[0,1)$. Also, we see that $\hat{G}(t)$ is continuous over $[0,1)$, $\hat{G}(0)=-1<0$ and $\lim_{t\rightarrow 1^-}\hat{G}(t)=+\infty$. Thus,  $\hat{G}(t)=0$ has a unique real root $\delta^*$ in   $(0,1)$. By noting that  $\delta_{k}\leq\delta_{2k}\leq\delta_{3k}<\delta^*$ and  \eqref{parti324},   we deduce that
$\hat{\varrho}\leq\hat{G}(\delta_{3k})+1<\hat{G}(\delta^*)+1=1.$
\end{proof}

\begin{Rem}\label{gmf-l2} \emph{While the main result in this paper is shown by considering the generalized mean function (\ref{G-M-fun}) satisfying the conditions of Lemma \ref{lem-GM-fun}, the error bound of the algorithm can be established with more general functions than those described by Lemma \ref{lem-GM-fun}. In fact, the inequality \eqref{rp_Omg_qk} is key to the establishment of Theorem \ref{main-thm}. While \eqref{rp_Omg_qk} is shown under the condition of Lemma \ref{lem-GM-fun}, we can verify that some other functions may also ensure the inequality \eqref{rp_Omg_qk}.
For instance, let us consider  the norm  $f(z)=\|z\|_\ell ~(\ell>1),$  where $z\in [0,1]^k$, which can be also viewed as a generalized mean function  $\Gamma_{\theta}(z)$ with $\theta=(1,\ldots,1)^T\in\mathbb{R}_{++}^{k}$ and $\Psi (t)=\phi_i (t) = t^\ell$ for $i=1,..., k.$ Since Hessian matrix $\nabla^2f(z)$ is discontinuous at $0$, so this function does not satisfy the conditions of Lemma \ref{lem-GM-fun} and thus the proof of Lemma \ref{lem-rp_Omg_qk} is not suitable for this function. However, for this case, $f\left(|r^p_{(q,k)}|/\|r^p_{(k,k)}\|_2\right)\geq\gamma f\left(|r^p_{(k,k)}|/\|r^p_{(k,k)}\|_2\right)$ is reduced to $\|r^p_{(q,k)}\|_\ell\geq\gamma \|r^p_{(k,k)}\|_\ell.$   Note that the  norms in $\mathbb{R}^k$ are equivalent in the sense that there exist two positive constants $c_2\geq c_1>0$ such that $c_1\|z\|_2\leq \|z\|_\ell\leq c_2\|z\|_2. $ This implies that \eqref{rp_Omg_qk} also holds for $g(\gamma)=\gamma\frac{c_1}{c_2}$. In particular, $g(\gamma)=\gamma$ when $f(z)=\|z\|_2. $  }
\end{Rem}

\section{Numerical experiments}\label{Num-exp}

In this section, we compare the numerical performances of six algorithms including  DTAM, PGROTP, NTP, StOMP, SP and OMP on solving several types of linear inverse problems including the recovery of synthetic sparse signal, reconstruction of natural audio signals as well as color image denoising. All experiments are performed on a PC with the processor Intel(R) Core(TM) i7-10700 CPU@ 2.90 GHz and 16 GB memory. The CVX \cite{GB17} with solver \emph{`Mosek'} \cite{AA20} was used to solve convex optimization subproblems involved in DTAM and PGROTP. We take \eqref{recov-criter} as  the stopping criterion  in Section \ref{sec-syn} in noiseless situations, and take $\|x^p-x^{p-1}\|_2/\|x^{p}\|_2\leq 10^{-3}$ in Sections \ref{sec-signal} and \ref{sec-image} in noisy settings.  The maximum numbers of iterations of DTAM, PGROTP, NTP and SP were set to be 50, 50, 150, 150, respectively, while OMP by its nature is performed exactly $k$ iterations. The generalized mean function $\Gamma_\theta(z)$ in DTAM is given by \eqref{Exam1-G-M} with $\sigma=1$ and $\theta=(1,\ldots,1)^T\in\mathbb{R}_{++}^{k}$ and the parameters $\gamma, \beta$ are set as $\gamma=0.1$ and $ \beta=0.4$, and unless otherwise specified, these parameters remained unchanged  throughout the experiments. The parameters $(\alpha, \lambda)$  in NTP are set as in \cite{ZL22}, i.e., $\alpha=5$ and $\lambda= 1. $ The number of stages of StOMP is set to be 50, and its threshold parameter $t_s$  is determined by the CFAR threshold selection rule \cite{DTDS12}.

\subsection{Experiments with synthetic data}\label{sec-syn}
We consider the recovery of a sparse vector $x^*$ from accurate measurements $y=Ax^*$ with $A=\hat{A}\cdot \textrm{diag}(1/\|\hat{A}_1\|_2,\ldots,1/\|\hat{A}_n\|_2)$, where  $x^*\in\mathbb{R}^{ n}$ and $\hat{A}\in\mathbb{R}^{m\times n}$ are randomly generated with $n=4000$ and $m=0.2n$, and $\hat{A}_i$'s are the columns of matrix $\hat{A}$.    Moreover, the nonzeros of $x^*$ and entries of $\hat{A}$ are  standard Gaussian random variables, and the position of nonzeros of $x^*$ follows the uniform distribution. We first compare the success frequencies and average runtime of these algorithms for solving 100 random examples of  $(A,x^*)$ for every given sparsity level $k,$ where $ k= 5+5j, j=1, ..., 71.$  In our experiments, the recovery is counted as `success' if the solution $x^p$ generated by an algorithm satisfies the criterion
\begin{equation}\label{recov-criter}
\|x^p-x^*\|_2/\|x^*\|_2\leq 10^{-3}.
\end{equation}
The experiment results are summarized in Fig. \ref{syn-data}. The first figure on the left indicates that the DTAM can achieve the success frequency (i.e., the ratio of the number of successes and the number of random examples) comparable to several existing methods and may outperform these existing methods on many examples.
  \begin{figure}[htbp]
    \begin{indented}
 \item[]
   \subfigure[Success frequency comparison]{
\begin{minipage}[t]{0.45\linewidth}
%\centering
  \includegraphics[width=1.05\textwidth,height=0.95\textwidth]{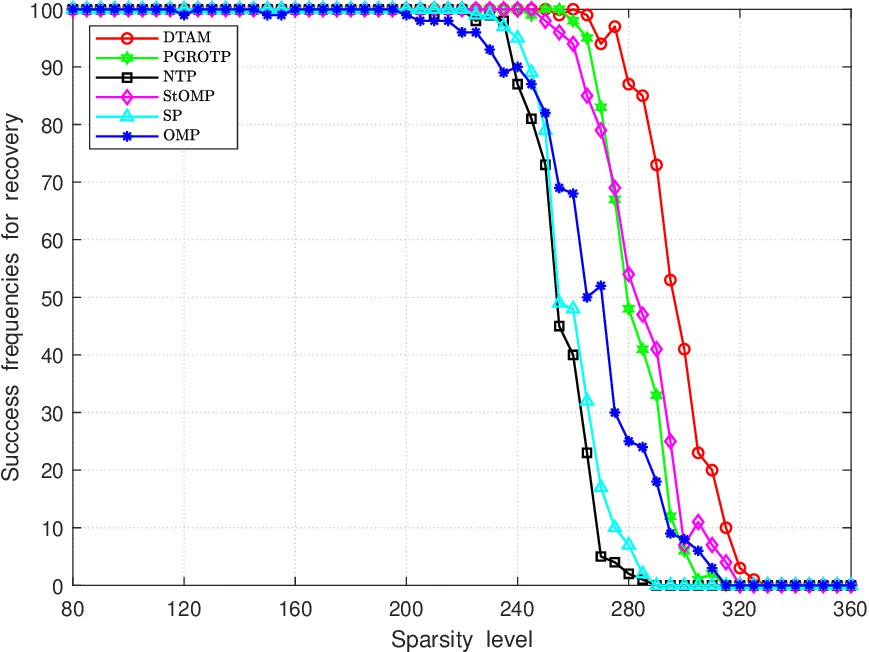}
   \end{minipage}
   }\hspace{0.3cm}
      \subfigure[Runtime]{
\begin{minipage}[t]{0.45\linewidth}
%\centering
  \includegraphics[width=1.05\textwidth,height=0.95\textwidth]{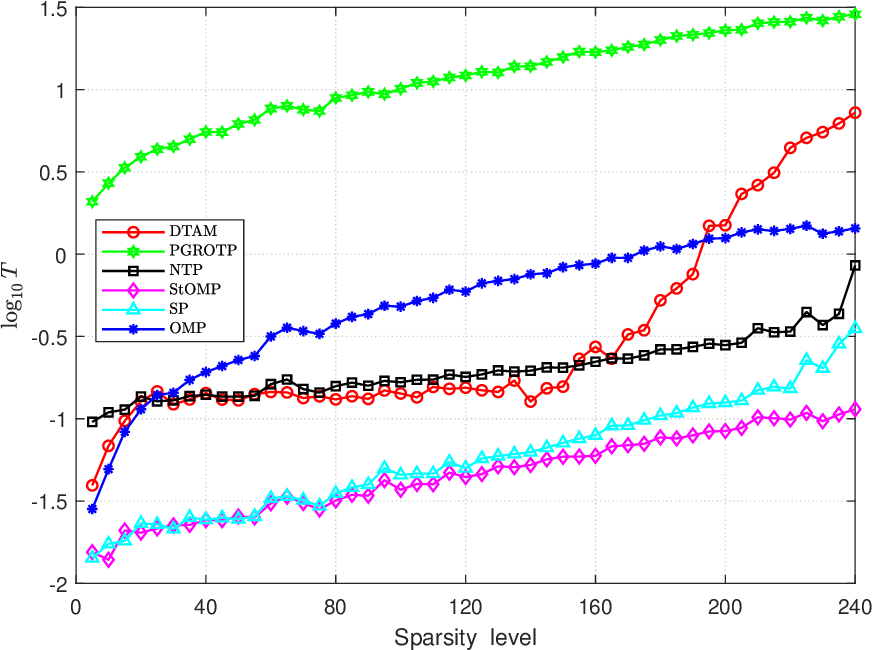}
   \end{minipage}
   }
    \end{indented}
   \caption{Comparison of  success frequencies and runtime on synthetic data, and $T$ is the average CPU time (in seconds) for recovery.}\label{syn-data}
   \end{figure}
In Fig. \ref{syn-data}(b), we use $T$ to denote the average CPU time required for these algorithms to recover sparse vectors. Clearly, DTAM works much faster than PGROTP (since DTAM solves the subproblem \eqref{alg-DSROTP-2} in a lower dimensional subspace, whose dimension is at most $2k$), while it  is slower than StOMP and SP.   Also, DTAM consumes less time than  NTP and OMP
for relatively small $k$, while it takes more time than  NTP for $k\geq m/5$ and OMP  for $k\geq m/4$.

\subsection{Reconstruction of audio signal }\label{sec-signal}

The first row in Fig. \ref{fig-signal} is   an audio signal $d\in\mathbb{R}^{ n}$  with  $n=16384$, which is the sound of an unknown Bird sampled at 48 kHz. We aim to reconstruct the bird signal from the accurate measurements $y=Bd$ with $B\in\mathbb{R}^{m\times n}$ being a normalized  Gaussian matrix given as in Section \ref{sec-syn}, wherein $m=\lceil \kappa \cdot n\rceil$ and $\kappa$ is the sampling rate.
Generating  a discrete wavelet matrix $\Phi\in\mathbb{R}^{n\times n}$ from the DWT with nine levels of  \emph{`sym16'} wavelet, the audio signal $d$ can be sparsely represented as $d=\Phi^T x$, where the wavelet coefficient vector $x\in\mathbb{R}^{ n}$ is $k$-compressible. Thus the reconstruction of $d$  from $y=Bd$ is transformed to  the recovery of a $k$-sparse vector $\hat{x}\in\mathbb{R}^{ n}$ from $y=A x$ with $A=B\Phi^T$ by using the model \eqref{model}, where $\hat{x}$ is  the best $k$-term approximation of $x$  and the sparsity level is set as   $k=\lceil 0.3m\rceil$. Once $\hat{x}$ is recovered by the algorithm, the reconstructed signal $\hat{d}\in\mathbb{R}^{ n}$ can be obtained by  $\hat{d}=\Phi^T \hat{x}$ immediately. The quality of reconstruction is evaluated by the SNR, defined as follows:
\begin{align*}
SNR:=20\cdot log_{10}(\|d\|_2/\|d-\hat{d}\|_2).
\end{align*}

\begin{figure}[htbp]
\begin{center}
   \begin{indented}
 \item[]
\begin{minipage}[t]{\linewidth}
%\centering
  \includegraphics[width=\textwidth,height=0.7\textwidth]{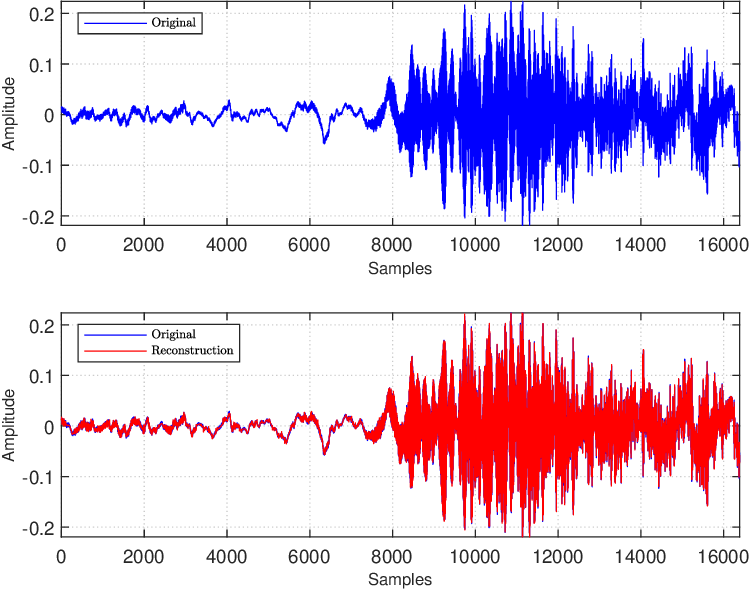}
   \end{minipage}
    \end{indented}
   \caption{Reconstruction of an  audio signal by DTAM with $\kappa=0.5;$ The first row is the original signal, and the second one demonstrates both the original signal ({\it blue}) and the reconstructed one by DTAM ({\it red}).}\label{fig-signal}
   \end{center}
   \end{figure}

The DTAM can successfully reconstruct the audio signal. This can be seen from Fig. \ref{fig-signal} that  the reconstructed signal (red) by DTAM with $\kappa=0.5$ is clearly matching with the original signal (blue). The performances of the algorithms with different sampling rates $\kappa=0.35,0.4,0.5$ are summarized in Table \ref{table-SNR-Time}. The second row of the table indicates that the SNRs of DTAM are almost the same as that of SP, and they are always larger than other four algorithms for each given $\kappa$. For instance, the SNR of DTAM exceeds that of StOMP by 0.76 dB  as $\kappa=0.5$, and by 2.3 dB  as $\kappa=0.35$. This means DTAM performs better on audio signal reconstruction than several algorithms except SP for small $\kappa$. The third row of Table \ref{table-SNR-Time} reveals that DTAM consumes less time for solving the problems than PGROTP and OMP for these given values of  $\kappa$, while it spends more time than other three.
   \begin{table}[h]
 \caption{Comparison of the SNR (dB) and CPU time (in seconds) for   algorithms with different sampling rates $\kappa$.}\label{table-SNR-Time}
 \vspace{0.2cm}
 \begin{indented}
 \item[]\begin{tabular}{cccccccc}
\hline
& $\kappa$ & DTAM &PGROTP&NTP&StOMP&SP&OMP\\
\hline
 \multirow{3}*{SNR}
 &0.35&	21.47 	&20.71 	&20.09 	&19.17 	&21.67 	&21.04 	\\
&0.4&	23.48 	&23.30 	&22.85 	&21.85 	&23.67 	&23.08 	\\
&0.5&	26.33 	&25.61 	&25.39 	&25.57 	&26.35 	&25.66 	 \\
\hline
  \multirow{3}*{Time}
&0.35	&516 	&1770 	&81 	&108 	&203 	&1044 	 \\
&0.4	&765 	&1819 	&92 	&91 	&196 	&1541 	 \\
&0.5	&1421 	&2897 	&103 	&138 	&826 	&3069 	\\
\hline
 \end{tabular}
 \end{indented}
 \end{table}

\subsection{Image denoising}\label{sec-image}
   We now demonstrate the performance of DTAM on color image denoising. Fig. \ref{denoising} (a) is the original image \emph{ShiGanLi} of size  $n\times n\times 3$ with $n=1024$, which is an ancient cooking vessel.  The Fig. \ref{denoising} (b) is the noised image obtained by adding  Salt and Pepper noise  with noise density 0.08  to the original image in Fig. \ref{denoising} (a), wherein Salt noise is added to the rows ranging from 1  to $\lfloor 0.8n\rfloor$ of the original  image while Pepper noise is added to the remaining rows.
    \begin{figure}[htbp]
   \begin{indented}
 \item[]\centering
      \subfigure[Original image]{
\begin{minipage}[t]{0.45\linewidth}
  \includegraphics[interpolate=true,width=\textwidth,height=\textwidth]{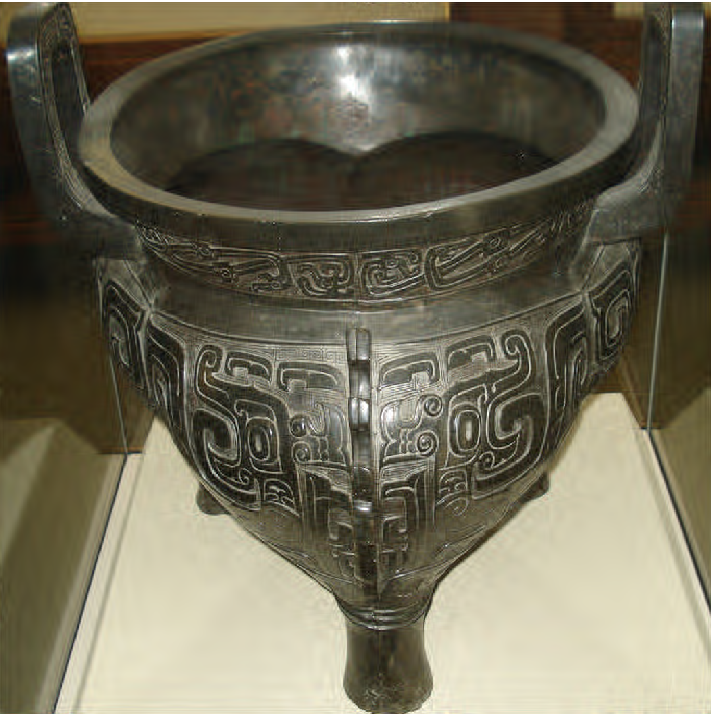}
   \end{minipage}
   }\\
   \subfigure[Noisy image]{
\begin{minipage}[t]{0.45\linewidth}
%\centering
  \includegraphics[interpolate=true,width=\textwidth,height=\textwidth]{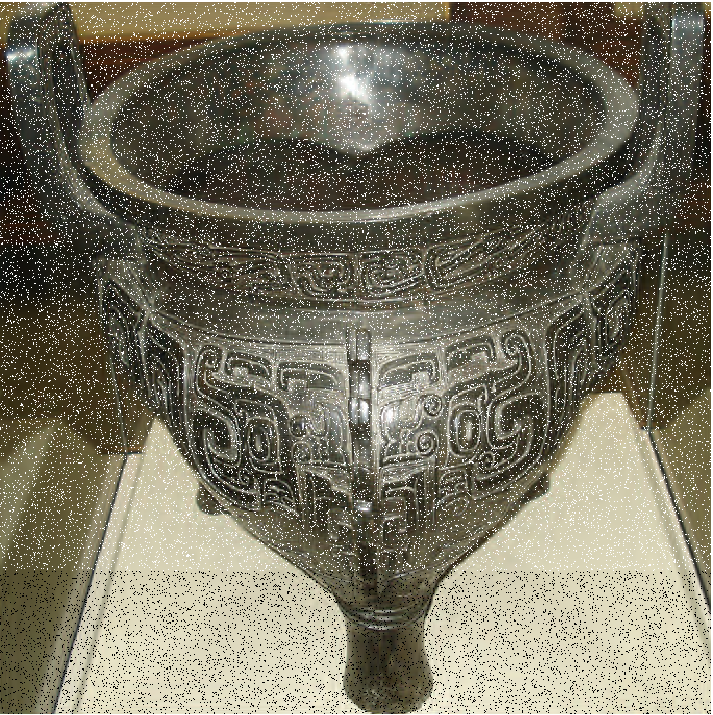}
   \end{minipage}
   }
   \subfigure[Denoised by DTAM]{
\begin{minipage}[t]{0.45\linewidth}
%\centering
  \includegraphics[interpolate=true,width=\textwidth,height=\textwidth]{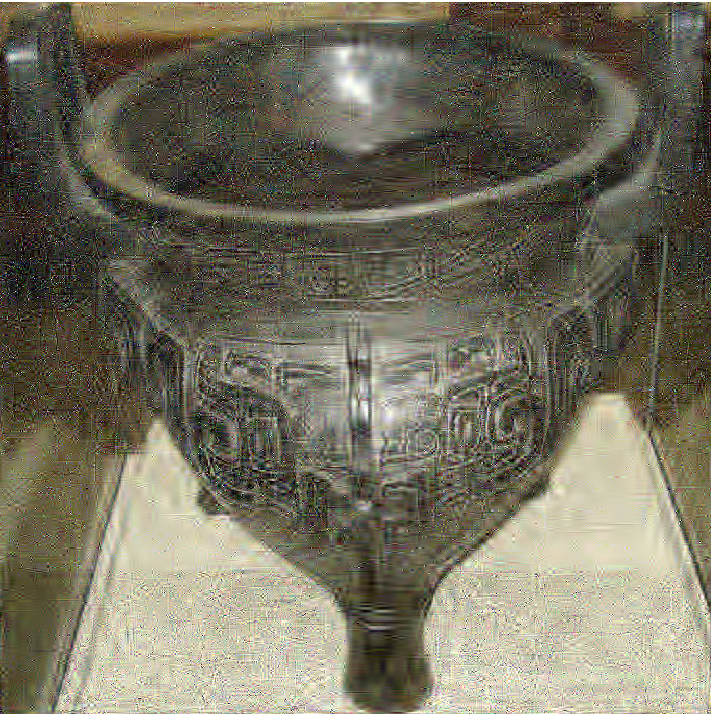}
   \end{minipage}
   }
   \end{indented}
   \caption{Performance of DTAM on image denoising.}\label{denoising}
   \end{figure}
For a given channel of the noisy image, the main steps for image denoising are as follows: First, perform a sparse representation of noisy image  via the DWT with five levels of  \emph{`sym16'} wavelet. Its coefficient matrix, denoted by $\tilde{X}$,  is compressible. Then, consider  the  accurate measurements $Z=A\tilde{X}$ of the coefficient matrix $\tilde{X}$, where $A$ is an $m\times n$ normalized Gaussian matrix given as in Section \ref{sec-syn} with  $m=\lceil 0.6 n\rceil.$  Finally, recover the  coefficient matrix  $X$ of the original image by an algorithm from the data $(A,Z)$, in which $Z$ can be seen as the inaccurate measurements of $X. $ After this, reconstruct the original image by using the inverse DWT.  Note that the above steps  need to   perform  three times due to three channels of the color image. The sparsity level is taken as $k=\lceil m/4\rceil$ for all algorithms.
The value of the parameter $\gamma$ in DTAM  is changed to 0.4, and other parameters remain unchanged.  Fig. \ref{denoising} (c) shows that DTAM can efficiently remove the noise and restore the image quality. We also use the following PSNR to evaluate the quality of denoised image:
\begin{align*}
PSNR:=20\cdot log_{10}(255/\sqrt{MSE}),
\end{align*} where $MSE$ is the mean-squared error  between the denoised image and  the original one. We only record the PSNR on the Y channel in YCbCr color space.  The PSNRs for several algorithms are displayed in Table \ref{table-PSNR}, from which we observe that the PSNRs  of these algorithms are close to each other while PGROTP is slightly better than other algorithms. These algorithms bring at least 5.1 dB improvement in PSNR compared to the noisy image.
 \begin{table}[hbtp]
 %\centering
 \caption{Comparison of  the PSNR (dB) for  algorithms.}\label{table-PSNR}
  \vspace{0.2cm}
 \begin{indented}
 \item[]
 \begin{tabular}{cccccccc}
\hline
&DTAM &PGROTP&NTP&StOMP&SP&OMP&Noisy image\\
\hline
PSNR&20.81 	&20.96 	&20.89 &	20.89 	&20.93 	&20.53& 	15.41\\
 \hline
 \end{tabular}
  \end{indented}
 \end{table}

\section{Conclusions}\label{conclusion}
In this paper, the algorithm DTAM is proposed for sparse linear inverse problems through merging a few algorithmic development techniques such as the sparse search direction, dynamic index selection and dimensionality reduction.  The computational complexity of DTAM is lower than that of ROTP-type algorithms. A unique feature of DTAM is that it employs a generalized mean function to facilitate a dynamic choice of the vector bases to construct the solution of linear inverse problems, and that the search direction in the algorithm is a linear combination of the negative gradients of error metric at the iterates produced by the algorithm. The error bound of DTAM has been established under suitable assumptions. Moreover, the error bound for the existing PGROTP method has also derived for the case $\bar{q}=k$ for the first time. Numerical experiments show that DTAM can compete with several existing algorithms, including PGROTP, NTP, StOMP, SP and OMP, in successfully locating the solution of linear inverse problems.

\noindent
\section*{Data availability} The real data are available from  \href{https://zhongfengsun.github.io/}{https://zhongfengsun.github.io/}.

\section*{Acknowledgments}
 This work was supported in part by the National Key R\&D Program of China (2023YFA1009302), the National Natural Science Foundation of China (12071307, 12371305 and 12371309), Guangdong Basic and Applied Basic Research Foundation (2024A1515011566),  Shandong Province Natural Science Foundation
(ZR2023MA020),  and  Domestic and Overseas  Visiting Program for  the Middle-aged and Young Key Teachers of  Shandong University of Technology. 

\section*{ORCID iDs}
Zhong-Feng Sun {\href{https://orcid.org/0000-0003-1794-7869}  {~https://orcid.org/0000-0003-1794-7869} }\\
 Yun-Bin Zhao {\href{https://orcid.org/0000-0002-2388-9047}  {~https://orcid.org/0000-0002-2388-9047} } \\
 Jin-Chuan Zhou {\href{https://orcid.org/0000-0002-2469-2123} {~https://orcid.org/0000-0002-2469-2123} } \\
Zheng-Hai, Huang {\href{https://orcid.org/0000-0003-2269-961X} {~https://orcid.org/0000-0003-2269-961X} }

\end{document}